\documentclass[journal,comsoc]{IEEEtran}

\usepackage{graphicx}
\usepackage{amssymb}
\usepackage{amsmath}
\usepackage{cite}
\usepackage{subfigure}
\usepackage{mathrsfs}
\usepackage[displaymath,mathlines]{lineno}
\usepackage{color}
\usepackage{tabulary}
\usepackage{multirow}
\usepackage{algpseudocode}
\usepackage{algorithm,algpseudocode}
\usepackage{algorithmicx}
\usepackage{pbox}
\usepackage{multicol}
\usepackage{lipsum}
\usepackage{amsthm}
\usepackage{relsize}
\usepackage{lipsum}
\usepackage{epstopdf}
\usepackage{mathtools}% http://ctan.org/pkg/mathtools
\usepackage{calc}% http://ctan.org/pkg/calc

\makeatletter
\newcommand{\doublewidetilde}[1]{{%
		\mathpalette\double@widetilde{#1}}}
\newcommand{\double@widetilde}[2]{%
		\sbox\z@{$\m@th#1\widetilde{#2}$}%
		\ht\z@=.5\ht\z@
		\widetilde{\box\z@}}
\makeatother
\newcommand*{\vtilde}[2][0pt]{% \vtilde[<lift>]{<stuff>}
	\setbox0=\hbox{$#2$}%
	\widetilde{\mathrlap{\phantom{\rule{\wd0}{\ht0+{#1}}}}\smash{#2}}%
}

\newcommand\numberthis{\addtocounter{equation}{1}\tag{\theequation}}

\newtheorem{theorem}{Theorem}
\newtheorem{lemma}{Lemma}
\newtheorem{corollary}{Corollary}

\newcounter{eqnback}
\newcounter{eqncnt}

\begin{document}
\title{Joint Pilot Design and Uplink Power Allocation in Multi-Cell Massive MIMO Systems}

\author{\normalsize Trinh Van Chien \textit{Student Member}, \textit{IEEE}, Emil Bj\"{o}rnson, \textit{Senior Member}, \textit{IEEE}, and Erik G. Larsson, \textit{Fellow}, \textit{IEEE}
\thanks{
        The authors are with the Department of Electrical
    Engineering (ISY), Link\"{o}ping University, SE-581 83 Link\"{o}ping,
    Sweden
        (email: trinh.van.chien@liu.se; emil.bjornson@liu.se; erik.g.larsson@liu.se).
}
\thanks{This paper was supported by the European Union's Horizon 2020 research and innovation programme under grant agreement No 641985 (5Gwireless). It was also supported by ELLIIT and CENIIT.}
\thanks{
	Parts of this paper were presented at IEEE ICC 2017.
}
}

\maketitle

\begin{abstract}
This paper considers pilot design to mitigate pilot contamination and provide good service for everyone in multi-cell Massive multiple input multiple output (MIMO) systems. Instead of modeling the pilot design as a combinatorial assignment problem, as in prior works, we express the pilot signals using a pilot basis and treat the associated power coefficients as continuous optimization variables. We compute a lower bound on the uplink capacity for Rayleigh fading channels with maximum ratio detection that applies with arbitrary pilot signals. We further formulate the max-min fairness problem under power budget constraints, with the pilot signals and data powers as optimization variables. Because this optimization problem is non-deterministic polynomial-time hard due to signomial constraints, we then propose an algorithm to obtain a local optimum with polynomial complexity. Our framework serves as a benchmark for pilot design in scenarios with either ideal or non-ideal hardware. Numerical results manifest that the proposed optimization algorithms are close to the optimal solution obtained by exhaustive search for different pilot assignments and the new pilot structure and optimization bring large gains over the state-of-the-art suboptimal pilot design.
\end{abstract}

% no keywords

% \fi
\IEEEpeerreviewmaketitle

\section{Introduction}
%\vspace*{-0.3cm}
The demands on capacity and reliability in wireless cellular networks are continuously increasing. It is known that multiple input multiple output (MIMO) techniques can improve both capacity and reliability \cite{Tse2005a,  Jiang2007a, Li2010b}, but current systems only support up to eight antennas per base station (BS). While codebook-based channel acquisition is attractive in such small-scale MIMO systems, these methods are not scalable and unable to support the fifth generation ($5$G)  demands on spectral efficiency (SE) in non-line-of-sight conditions \cite{Bjornson2016b}. Massive MIMO was proposed in \cite{Marzetta2010a} as a possible solution and it has emerged as a key $5$G technology, because it offers significant improvements in both SE and energy efficiency \cite{Marzetta2016a,Marzetta2010a, Bjornson2016b, Wang2016a, Mohammed2014a}. By equipping the BSs with hundreds of antennas, mutual interference, thermal noise, and small-scale fading can be almost eliminated by virtue of the channel hardening and favorable propagation properties \cite{Marzetta2016a}. The BSs only need to use linear detection schemes, such as maximum ratio (MR) or zero forcing, to achieve nearly optimal performance \cite{Chien2017a}. In addition, the SE only depends on the large-scale fading coefficients, thus power control algorithms are easier to deploy than in small-scale MIMO systems, which are greatly affected by small-scale fading~\cite{Chien2016b}.

The uplink (UL) detection and downlink precoding in Massive MIMO  are based on instantaneous channel state information (CSI), which the BSs obtain from UL pilot signals. Mutually orthogonal pilots are desirable, but this is impractical in multi-cell scenarios  since the pilot overhead would be proportional to the total number of users in the entire system. The consequence is that the pilot signals need to be reused across cells. This leads to pilot contamination \cite{Jose2011b, Bogale2014}, where users sending the same pilot degrade each others channel estimation and cause large mutual interference. Hence, the pilot design is of key importance in Massive MIMO and should be optimized to mitigate the pilot contamination effects.

The baseline scheme for mitigating pilot contamination is to introduce a pilot reuse factor $f$, such that each pilot is only reused in $1/f$ of the cells. 
This approach, which was studied in \cite{Yang2015a,Li2015a,You2015a,Bjornson2016a}, can greatly reduce the pilot contamination, even if the pilots are randomly assigned within each cell. However, this gain comes at the cost of using $f$ times more pilots than in a system reusing the pilots in every cell. 
For any given cell, only a few users in the neighboring cells cause most of the potential pilot contamination, thus it is most important that these potential contaminators are assigned different pilots from the users in the given cell. Algorithms for coordinated pilot assignment were proposed in \cite{Jin2015a,Zhu2015a, Ahmadi2016a, Xu2015a}. A pilot reuse dictionary was defined in \cite{Jin2015a} and the corresponding pilot assignment problem was shown to be non-deterministic polynomial-time hard (NP-hard), which motivates the design of heuristic assignment mechanisms. Although \cite{Jin2015a} proposed several greedy algorithms, the optimized SE was far from that with exhaustive search over all pilot assignments.
Graph theory was used for pilot assignment in \cite{Zhu2015a}, by exploiting variations in the large-scale fading coefficients.
A method called ``smart pilot assignment'' was proposed in \cite{Xu2015a} to enhance the max-min fairness SE level, by optimizing a heuristic mutual interference metric.
Alternatively, \cite{Ahmadi2016a} formulated the pilot assignment problem as a potential game.
The numerical results in \cite{Zhu2015a, Ahmadi2016a, Xu2015a} show performance that is similar to an exhaustive search, but with a substantially lower computational complexity. Moreover, the authors of \cite{Yin2016a,You2016a} utilized particular channel properties to reduce channel estimation errors and mitigate pilot contamination. In particular, \cite{Yin2016a} utilized the orthogonality among different channels and an assumed low-rankness of the channel covariance matrices. An adjustable phase shift pilot construction was suggested in \cite{You2016a} based on the relationship between channel correlations in the frequency domain and their power angle-delay spectrum. However, all these algorithms rely on the assumption of fixed pilot and data power.

The pilot and  payload data powers are usually treated as constants in the Massive MIMO literature, but it is known from \cite{Zhao2013, Bogale2014} that the performance can be much improved by using the optimal power allocation, which balances the mutual interference levels. To improve the channel estimation quality, more power might also be assigned to the pilots than to the data transmissions \cite{Victor2017a, Guo2014a}. For single-cell systems, \cite{Victor2017a} showed that a pilot-data power imbalance is especially important for cell-edge users. Moreover, the power allocation that maximizes the sum SE is much different from the one that maximizes the max-min SE. Similar behaviors for multi-cell systems were observed in \cite{Guo2014a}. The authors in \cite{Liu2016a} considered power optimization problems with pilot reuse factors. To the best of our knowledge, no prior work analyzes joint pilot design and power control in Massive MIMO systems.

In this paper, we propose a novel pilot design and optimize the UL performance in multi-cell Massive MIMO systems, using the max-min fairness utility. Our main contributions are: 
\begin{itemize}
	\item We propose a new pilot design where the pilot signals are treated as continuous variables. We demonstrate that previous pilot designs are special cases of our proposal.
	
	\item Based on the proposed pilot design, we derive closed-form expressions of the SE with Rayleigh fading channels and MR detection, for the cases of ideal hardware and with hardware impairments. These expressions explicitly demonstrate how the SE is affected by mutual interference, noise, and pilot contamination.
	
	\item We formulate the max-min fairness problem for the proposed pilot design, by treating the pilot signals, pilot powers, and data powers as optimization variables. This is an NP-hard signomial program, so we propose an algorithm that finds a local optimum in polynomial time. For comparison the optimal solution by an exhaustive search of different pilot assignments is also investigated.
	
	\item The proposed algorithms are evaluated numerically, with either ideal hardware or  hardware impairments. The results show that our local solution is close to the global optimum by exhaustive search over different pilot assignments and demonstrate significant improvements over the heuristic algorithms in prior works.
	
\end{itemize}

A preliminary version of this work, focusing only on pilot optimization with fixed data powers, was presented in \cite{Chien2017b}.

The rest of this paper is organized as follows: Section \ref{Section: System Model} presents our proposed pilot structure and compares it with prior works. Lower bounds on the UL ergodic SE for arbitrary pilots are derived in Section \ref{Section: ULTransmission}, while Section \ref{Section: OptProblem} formulates the max-min fairness optimization problems and provides the global and local solutions. Sections \ref{section: Hardware_Impairments} and \ref{Section:CorrUncorr} extend our research to the case of hardware impairments and correlated Rayleigh fading, respectively. Finally, Section \ref{Section: Experimental Result} gives extensive numerical results and some conclusions are provided in Section \ref{Section: Conclusion}.

\textit{Notations}: Lower bold letters are used for vectors and upper cases are for matrices. $(\cdot)^T$ and $(\cdot)^H$ stand for regular transpose and Hermitian transpose, respectively. The superscript $\ast$ denotes the conjugate transpose of a complex number. $\mathbf{I}_{n}$ is the identity matrix of size $n \times n$. $\mathbb{C}^{m \times n} \; (\mathbb{R}^{m \times n})$ is the space of complex (real) $m\times n$ matrices, while $\mathbb{C}^{\tau_p}$ denotes the space of $\tau_p$-length complex vectors. $\mathbb{R}_{+}$ is the set of nonnegative real numbers. $\mathbb{E} \{\cdot \}$ denotes the expectation of a random variable and $\| \cdot \| $ is the Euclidean norm. Finally, $\mathcal{CN} (\cdot, \cdot)$ is the circularly symmetric complex Gaussian distribution, while $\mathcal{N} (\cdot, \cdot)$ is the normal distribution.

\section{Pilot Designs for Massive MIMO Systems} \label{Section: System Model}
%\vspace*{-0.3cm}
We consider the UL of a multi-cell Massive MIMO system with $L$ cells.  Each cell consists of a BS equipped with $M$ antennas that serves $K$ single-antenna users. All tuples of cell and user indices belong to a set $\mathcal{S}$ defined as
\begin{equation}
\mathcal{S} = \left\{ (i,t): \; i \in \{ 1, \ldots, L\}, \; t \in \{ 1, \ldots, K \} \right \}.
\end{equation}
The radio channels vary over time and frequency. We divide the time-frequency plane into coherence intervals, each containing $\tau_c$ samples, such that the channel between each user and each BS is static and frequency flat. In each coherence block, the pilot signaling utilizes $\tau_p$ symbols and the remaining is dedicated to data transmission. In this paper, we focus on the UL, so the fraction $(1 - \tau_p / \tau_c)$ of the coherence interval is dedicated to UL data transmission. However, it is straightforward to extend our work to the downlink by using time division duplex (TDD) and channel reciprocity. We assume $1 \leq \tau_p < \tau_c$ to keep the training process feasible and stress that the case $\tau_p < KL$ is of practical importance since it gives rise to pilot contamination and since $L$ is large in practice.

\subsection{Proposed Pilot Design} \label{subsection: ProposedPilot}
%\vspace*{-0.1 cm}
 Let us denote the  $\tau_p$ mutually orthonormal basis vectors $\{\pmb{\phi}_1, \ldots, \pmb{\phi}_{\tau_p} \}$, where $\pmb{\phi}_b \in \mathbb{C}^{\tau_p}$ is a vector whose $b$th element has unit magnitude, and all other elements are equal to zero. The corresponding basis matrix is
\begin{equation}
\pmb{\Phi} = [\pmb{\phi}_1, \ldots, \pmb{\phi}_{\tau_p}].
\end{equation}
We assume that the pilot signals of the users can span arbitrarily over the above $\tau_p$ basis vectors. We aim at designing a pilot signal collection $ \{ \pmb{\psi}_{1,1}, \ldots, \pmb{\psi}_{L,K}\}$ comprising the $KL$ pilot signals used by all users in the network and each of them has the length of $\tau_p$ symbols.
The pilot signal of user~$k$ in cell~$l$ is $\pmb{\psi}_{l,k} = [\psi_{l,k}^1, \ldots, \psi_{l,k}^{\tau_p}]^T \in \mathbb{C}^{\tau_p}$ and the power that this user assigns to the $b$th pilot basis is denoted as $\hat{p}_{l,k}^b \geq 0$. Thus, the pilot of user~$k$ in cell~$l$ is
\begin{equation} \label{eq: ProposedPilotSequence}
\pmb{\psi}_ {l,k}=  \sum_{b =1}^{\tau_p} \sqrt{ \hat{p}_{l,k}^b } \pmb{ \phi}_{b}, \forall l,k.
\end{equation}
We stress that the pilot construction in \eqref{eq: ProposedPilotSequence}  can be used to create any set of $\tau_p$ orthogonal pilot signals (up to a unitary transformation) and many different sets of non-orthogonal signals. \footnote{The pilot signals in \eqref{eq: ProposedPilotSequence} are formed as linear combinations of basis vectors in the complex field. The new pilot design allows the use of nonorthogonal pilot signals even within a cell in order to get extra degrees of freedom to minimize the interference in the network.}   
The total pilot power consumption utilized by user~$k$ in cell~$l$ is $\| \pmb{\psi}_{l,k} \|^2 = \sum_{b=1}^{\tau_p} \hat{p}_{l,k}^b$  and we assume that it satisfies the power constraint
\begin{equation} \label{eq:Max-Power}
\frac{1}{\tau_p} \sum_{b=1}^{\tau_p} \hat{p}_{l,k}^b \leq P_{\textrm{max},l,k}, \forall  l,k,
\end{equation}
where $P_{\textrm{max},l,k}$ is the maximum pilot power for user~$k$ in cell~$l$. The inner product of two pilot signals $\pmb{\psi}_ {l,k}$ and $\pmb{\psi}_ {i,t}$ is
\begin{equation} \label{eq: Orthogonal_Property}
\begin{split}
\pmb{\psi}_ {l,k}^H \pmb{\psi}_ {i,t} = \sum_{b =1}^{\tau_p} \sqrt{ \hat{p}_{l,k}^b \hat{p}_{i,t}^b  }.
\end{split}
\end{equation}
These pilot signals are orthogonal if the product is zero, which only happens when they allocate their powers to different subsets of basis vectors. Otherwise, they are non-orthogonal and then the two users cause pilot contamination to each other. If the square roots of the powers allocated to the $K$ users in cell~$l$ are gathered in matrix form as
\begin{equation} \label{eq: ProposedPilotPower}
\pmb{P}_l = \begin{bmatrix}
\sqrt{\hat{p}_{l,1}^1}      & \sqrt{\hat{p}_{l,2}^1}     &  \cdots     &  \sqrt{\hat{p}_{l,K}^1} \\
\sqrt{\hat{p}_{l,1}^2}       &  \sqrt{\hat{p}_{l,2}^2}     & \cdots    & \sqrt{\hat{p}_{l,K}^2} \\
\vdots & \vdots &  \ddots      &  \vdots  \\
\sqrt{\hat{p}_{l,1}^{\tau_p}}      &  \sqrt{\hat{p}_{l,2}^{\tau_p}}     &  \cdots   &  \sqrt{\hat{p}_{l,K}^{\tau_p}}
\end{bmatrix} \in \mathbb{R}_{+}^{\tau_p \times K},
\end{equation}
 then the users in cell~$l$ utilize a pilot matrix defined as
\begin{equation} \label{eq: PilotStructure1}
\mathbf{\Psi}_l = [ \pmb{\psi}_{l,1}, \ldots,  \pmb{\psi}_{l,K} ] = \pmb{\Phi} \pmb{P}_l.
\end{equation}
We now describe the difference between this new pilot structure and the prior works, for example \cite{Xu2015a,Zhu2015a, Victor2017a, Guo2014a}. 

\subsection{Other Pilot Designs}
%\vspace*{-0.2 cm}
The works \cite{Xu2015a,Zhu2015a} considered the assignment of $\tau_p$ orthogonal pilot signals under the assumption of fixed equal pilot power. Using our notation, the pilot matrix in cell~$l$ is
\begin{equation} \label{eq: fixedPilotPower}
\widehat{\pmb{\Psi}}_l = [\hat{\pmb{\psi}}_{l,1}, \ldots, \hat{\pmb{\psi}}_{l,K} ] = \sqrt{\tilde{p}}  \pmb{\Phi} \pmb{\Pi}_l,
\end{equation}
where $0 < \tilde{p}  \leq \tau_p P_{\max,l,k}$ is the equal power level of all users. $\pmb{\Pi}_l \in \mathbb{R}_{+}^{\tau_p \times K}$ is a permutation matrix, that assigns the pilot signals to each user in cell~$l$. The assignment is optimized in \cite{Xu2015a,Zhu2015a} to minimize a heuristic mutual interference metric. Note that these works assume orthogonal pilot signals and equal power allocation, which are simplifications compared to \eqref{eq: PilotStructure1}. These assumptions are generally suboptimal. Apart from this, the selection of the optimal permutation matrices for cell~$l$ is a combinatorial problem, so to limit the computational complexity \cite{Xu2015a,Zhu2015a} and the references therein only study the special case of $\tau_p = K$.

The previous work \cite{Victor2017a} optimized the pilot powers to maximize functions of the SE, but the paper only considered a single cell without pilot contamination. The authors of \cite{Guo2014a} optimized the pilot powers to minimize the UL transmit power for a multi-cell system. This work assumed $\tau_p = K$ and a fixed pilot assignment. If $\tilde{p}_{l,k}$ is the pilot power of user~$k$ in cell~$l$, the square root of the power matrix allocated to the $K$ users in cell~$l$ is a diagonal matrix defined as
 \begin{equation}
\widetilde{\pmb{P}}_l =  \mathrm{diag} \left(\sqrt{\tilde{p}_{l,1}} , \ldots, \sqrt{\tilde{p}_{l,K}} \right),
 \end{equation}
where $ \mathrm{diag}(\mathbf{x})$ denotes the diagonal matrix with the vector $\mathbf{x}$ on the diagonal. The pilot matrix in cell~$l$ is then formulated as
\begin{equation} \label{eq: PilotStructure2}
\widetilde{\mathbf{\Psi}_l} = \pmb{\Phi} \widetilde{\pmb{P}}_l.
\end{equation}
Similar to \eqref{eq:Max-Power}, the pilot power at user~$k$ in cell~$l$ is limited as
\begin{equation}
 0 \leq \tilde{p}_{l,k} \leq \tau_p P_{\max,l,k}.
\end{equation}
Since orthogonal pilots and fixed pilot assignment are assumed, this is also a special case of \eqref{eq: PilotStructure1}. We can combine the pilot structure in \eqref{eq: PilotStructure2} and the idea of selecting a permutation matrix in \eqref{eq: fixedPilotPower} to jointly optimize the power allocation and pilot assignment. In particular, the pilot signals of the users in cell~$l$ are now defined as
\begin{equation} \label{eq: PilotStructure3}
 \breve{\mathbf{\Psi}}_l =  [\breve{\pmb{\psi}}_{l,1}, \ldots, \breve{\pmb{\psi}}_{l,K} ] = \pmb{\Phi}  \pmb{\Pi}_l \widetilde{\pmb{P}}_l.
\end{equation}
This modified pilot design is a special case of \eqref{eq: PilotStructure1} and has not been studied in prior works, but will be considered herein. In order to analyze the channel estimation, we define a pilot reuse set $\mathcal{P}_{l,k}$ including all tuples of cell and user indices that cause pilot contamination to user~$k$ in cell~$l$:
\begin{equation} \label{eq:ReuseSet1}
\mathcal{P}_{l,k} = \{ (i,t) \in \mathcal{S} : \;  \breve{\pmb{\psi}}_{i,t}^H \breve{\pmb{\psi}}_{l,k} \neq 0 \}.
\end{equation}
 We stress that designing an exhaustive search to obtain the best pilot assignment strategy is extremely computationally expensive.\footnote{For the first user in the first cell $(l=1, k=1 )$, there are $(K!)^{L-1}$ possibilities of $\mathcal{P}_{1,1}$. There are then $(K-1)!^{L-1}$ possible $\mathcal{P}_{1,2}$ and so on.} As in prior works, we only consider the case $\tau_p = K$ when using \eqref{eq: PilotStructure3} and we further assume that orthogonal pilots are used within each cell; that is, $\mathcal{P}_{l,k} \cap \mathcal{P}_{l,k'} = \emptyset$ for any user indices $k \neq k'$ in cell~$l$. To perform an exhaustive search, we need to construct a dictionary $\mathcal{D}$, see Fig.~\ref{Fig1Dic}, with all the possible combinations of pilot assignments in the network. Let $\chi_l^k \in \{1, \ldots, K \}$ denote the index of the pilot signal assigned to user~$k$ in cell~$l$. It follows that $\chi_l^k \neq\chi_l^{k'}$ for $k \neq k'$ since all users within a cell use different pilots.
The pilot assignment matrix $\mathbf{A} \in \{ 1, \ldots, K \}^{L \times K}$ containing the pilot indices of the $KL$ users is 
\begin{equation} \label{eq: matrixA}
\mathbf{A} = \begin{bmatrix}
\chi_1^1      &   \chi_1^2    &  \cdots     &  \chi_1^K \\
\chi_2^1      &  \chi_2^2     & \cdots    & \chi_2^K \\
\vdots & \vdots &  \ddots      &  \vdots  \\
\chi_L^1      & \chi_L^2    &  \cdots   &  \chi_L^K
\end{bmatrix} .
\end{equation}
Each row of $\mathbf{A}$ contains $1$ to $K$ and there are $K!$ different combinations, each defining a permutation matrix $\pmb{\Pi}_l$ for the pilot signals in \eqref{eq: fixedPilotPower} and \eqref{eq: PilotStructure3}. The dictionary $\mathcal{D} \overset{\Delta}{=} \{ \mathbf{A} \}$ contains all the $(K!)^L$ pilot assignment matrices. For each $\mathbf{A} \in \mathcal{D}$, we can extract the pilot reuse sets $\mathcal{P}_{l,k}^{\mathbf{A}}, \forall l,k$ as\footnote{Each collection $\{ \mathcal{P}_{l,k}^{\mathbf{A}} \}$ of pilot reuse sets is generated by $K!$ different $\mathbf{A}  \in \mathcal{D}$. By eliminating the $K!-1$ copies, the size of the dictionary $\mathcal{D}$ can be reduced to $(K!)^{L-1}$, which still grows rapidly with $K$ and $L$.}
 \begin{equation} \label{eq:ReuseSet2}
\mathcal{P}_{l,k}^{\mathbf{A}} = \{  (i,t) \in \mathcal{S}: \;  \chi_{l}^k = \chi_{i}^t \}. %\chi_{l}^k, \chi_{i}^t \in \mathbf{A},
%\mathcal{P}_{l,k} = \{ (i,t): \; (i,t) \in \mathcal{S},\chi_{l}^k, \chi_{i}^t \in \mathbf{A}, \chi_{l}^k = \chi_{i}^t \}.
\end{equation}
The dictionary $\mathcal{D}$ %with the two main features in $\eqref {eq: matrixA}$ and $\eqref {eq:ReuseSet2}$ 
will be later used to obtain the pilot assignment that maximizes the SE performance.
 %\begin{table}
 %\caption{Pilot assignment indices for all users in the $L$ cells.} \label{Table1}
 %\begin{center}
 
 % \begin{tabular}{ |l|c|c|c|c| } 
 %	\hline
  % 	  & user~$1$ & user~$2$ & \ldots  & user~$K$ \\ 
   %	\hline
   %	cell~$1$ & $\chi_{1}^1$ & $\chi_{1}^2$ & \ldots &  $\chi_{1}^{K}$  \\ 
   %	\hline
   %	cell~$2$ & $\chi_{2}^1$ & $\chi_{2}^2$ & \ldots & $\chi_{2}^{K}$ \\ 
   %	\hline
   %	\ldots & \ldots & \ldots & \ldots &  \ldots \\ 
   %	\hline
   %	cell~$L$ & $\chi_{L}^1$ &$\chi_{L}^2$ & \ldots & $\chi_{L}^K$ \\ 
   %	\hline
 %\end{tabular}
 %\end{center}
 	
 %\end{table}
%The pilot reuse set $\mathcal{P}_{l,k}$ defined here is more general than, for example, in \cite{Chien2016b}. Specifically, for the cells utilizing same frequency band, \cite{Chien2016b} assumes the users with the same index deploy the same mutually orthogonal pilot sequences. Such heuristic supposition does not require the exhaustive search, but system performance can be far from what we expected since the best $\pmb{\Pi}_l$ may be not obtained for hierarchical networks with an overwhelming probability.
\begin{figure}[t]
	\centering
	\includegraphics[trim=20.5cm 54.5cm 3.2cm 9cm, clip=true, width=3.3in]{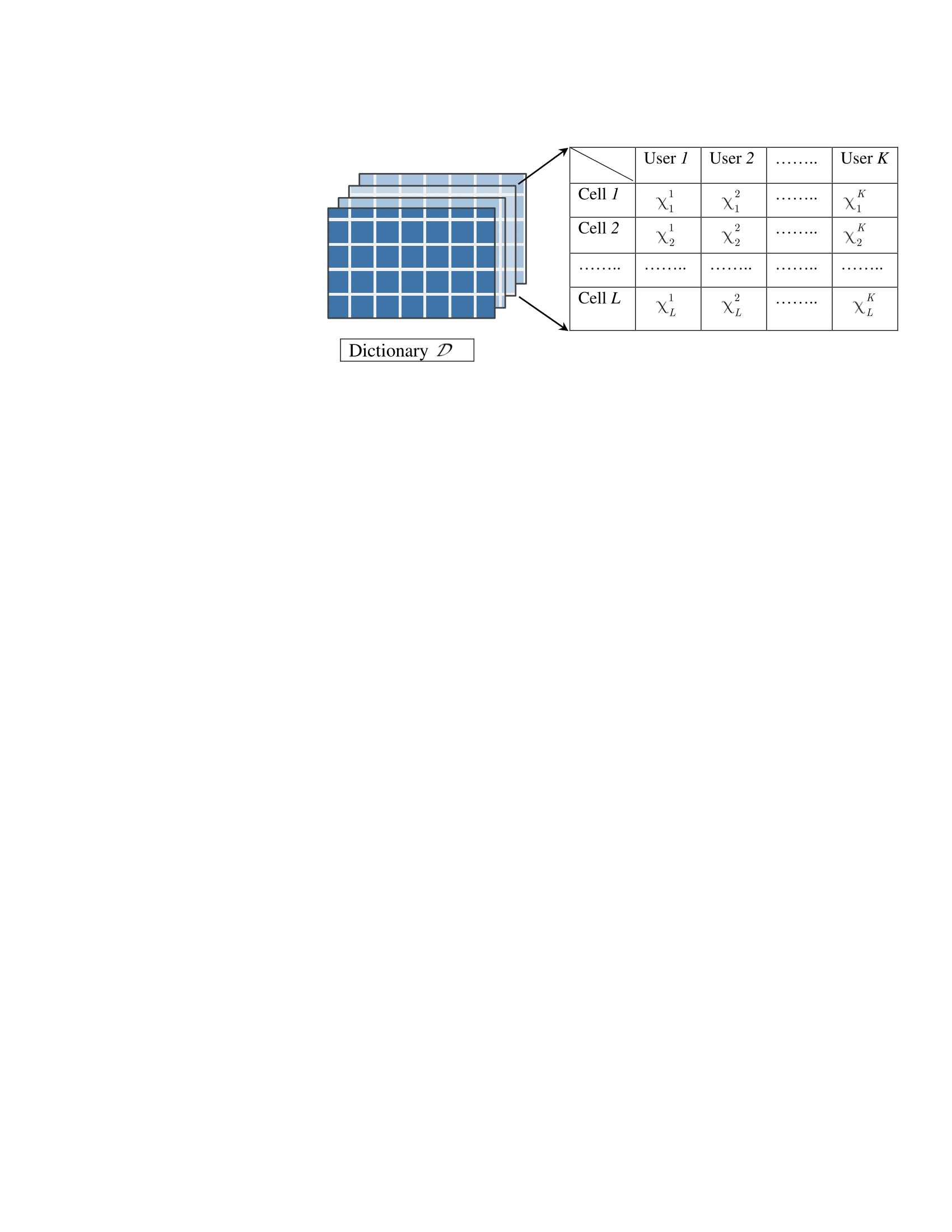}
	 %\vspace*{-0.2cm}
	\caption{The dictionary $\mathcal{D}$ contains all possible pilot assignment indices for all users in the network.}
	\label{Fig1Dic}
	%\vspace*{-0.2cm}
\end{figure}
\section{Uplink Massive MIMO Transmission} \label{Section: ULTransmission}
This section provides ergodic SE expressions with arbitrary pilot signals, which are later used for pilot optimization.
 
\subsection{Channel Estimation with Arbitrary Pilots}
%\vspace*{-0.3cm}
During the UL pilot transmission, the received signal $\mathbf{Y}_l \in \mathbb{C}^{M \times \tau_p}$ at the BS of cell~$l$ is 
\begin{equation} \label{eq: Received_Pilot}
\mathbf{Y}_l  = \sum_{(i,t) \in \mathcal{S}} \mathbf{h}_{i,t}^{l} \pmb{\psi}_{i,t}^{H} +  \mathbf{N}_l,
\end{equation}
where $\mathbf{h}_{i,t}^l \in \mathbb{C}^M$ denotes the channel between user~$t$ in cell~$i$ and BS~$l$. $\mathbf{N}_l \in \mathbb{C}^{M \times \tau_p}$ is the additive noise with independent elements distributed as $\mathcal{CN}(0, \sigma^2)$. Correlating $\mathbf{Y}_l$ in \eqref{eq: Received_Pilot} with the pilot $\pmb{\psi}_{l,k}$ of user~$k$ in cell~$l$, we obtain
\begin{equation} \label{eq:receivedPilotSiglk}
\mathbf{y}_{l,k} = \mathbf{Y}_{l} \pmb{\psi}_{l,k} = \sum_{(i,t) \in \mathcal{S}} \mathbf{h}_{i,t}^{l} \pmb{\psi}_{i,t}^H \pmb{\psi}_{l,k} + \mathbf{N}_{l} \pmb{\psi}_{l,k}.
\end{equation}
We consider uncorrelated Rayleigh fading since results obtained with this tractable model well matches the results obtained in non-line-of-sight measurements \cite{Gao2015a}. The channel between user~$t$ in cell~$i$ and BS~$l$ is distributed as
\begin{equation}
\mathbf{h}_{i,t}^l \sim \mathcal{CN} \left(  \mathbf{0} ,  \beta_{i,t}^l \mathbf{I}_M \right),
\end{equation}
where the variance $\beta_{i,t}^l$ determines the large-scale fading, including geometric attenuation and shadowing. %We note that $\beta_{i,t}^l$ is a deterministic constant. 
By using minimum mean squared error (MMSE) estimation, the distributions of the channel estimate and estimation error when using the pilot structure in \eqref{eq: PilotStructure1} are given in Lemma~\ref{lemma: Distribution}.
\begin{lemma} \label{lemma: Distribution}
If the system uses the pilot structure in \eqref{eq: PilotStructure1}, the MMSE estimate of $\mathbf{h}_{l,k}^l$ based on $\mathbf{y}_{l,k}$ in \eqref{eq:receivedPilotSiglk} is computed as
\begin{equation} \label{eq: ChannelEstimate}
\hat{\mathbf{h}}_{l,k}^l = \frac{\beta_{l,k}^l \sum\limits_{b=1}^{\tau_p} \hat{p}_{l,k}^b }{ \sum\limits_{(i,t) \in \mathcal{S} } \beta_{i,t}^l \left( \sum\limits_{b=1}^{\tau_p}  \sqrt{\hat{p}_{i,t}^b \hat{p}_{l,k}^b}  \right)^2 + \sigma^2  \sum\limits_{b=1}^{\tau_p} \hat{p}_{l,k}^b } \mathbf{y}_{l,k}.
\end{equation}
The channel estimate is distributed as
\begin{equation}
\hat{ \mathbf{h} }_{l,k}^l \sim \mathcal{CN} \left( \mathbf{0}, \gamma_{l,k}^l   \mathbf{I}_M  \right),
\end{equation}
where
\begin{equation}
\gamma_{l,k}^l = \frac{ (\beta_{l,k}^l)^2 \left(\sum\limits_{b=1}^{\tau_p} \hat{p}_{l,k}^b \right)^2 }{ \sum\limits_{(i,t) \in \mathcal{S} } \beta_{i,t}^l \left( \sum\limits_{b=1}^{\tau_p} \sqrt{\hat{p}_{i,t}^b  \hat{p}_{l,k}^b} \right)^2  +  \sigma^2 \sum\limits_{b=1}^{\tau_p} \hat{p}_{l,k}^b }.
\end{equation}
The estimation error $\mathbf{e}_{l,k}^l  = \mathbf{h}_{l,k}^l - \hat{\mathbf{h}}_{l,k}^l$ is independent of the channel estimate and distributed as
\begin{equation}
\mathbf{e}_{l,k}^l
\sim \mathcal{CN} \left( \mathbf{0} , \left( \beta_{l,k}^l - \gamma_{l,k}^l \right)\mathbf{I}_M  \right).
\end{equation}

\end{lemma}
\begin{proof}
The proof follows directly from standard MMSE estimation techniques in \cite{Kay1993a}.
\end{proof}
Lemma~\ref{lemma: Distribution} provides the MMSE estimator for the pilot design in \eqref{eq: PilotStructure1}. The pilot powers as well as inner products between pilot signals appear explicitly in the expressions. We now compute the channel estimate and estimation error of $\mathbf{h}_{l,k}^l$ when using the pilot structure in \eqref{eq: PilotStructure3}.
\begin{corollary} \label{co: Distribution}
If the system uses the alternative pilot structure in \eqref{eq: PilotStructure3}, the MMSE channel estimate in \eqref{eq: ChannelEstimate} is simplified to
\begin{equation}
\hat{\mathbf{h}}_{l,k}^l = \frac{\beta_{l,k}^l  }{ \sum\limits_{(i,t) \in \mathcal{P}_{l,k}^{\mathbf{A}} } \beta_{i,t}^l \tilde{p}_{i,t}  + \sigma^2  } \mathbf{y}_{l,k}.
\end{equation}
The estimate channel and estimation error are distributed as
\begin{equation}
\hat{ \mathbf{h} }_{l,k}^l \sim \mathcal{CN} \left( \mathbf{0}, \frac{ (\beta_{l,k}^l)^2  \tilde{p}_{l,k} }{ \sum\limits_{(i,t) \in \mathcal{P}_{l,k}^{\mathbf{A}}} \beta_{i,t}^l \tilde{p}_{i,t}  +  \sigma^2}   \mathbf{I}_M  \right),
\end{equation}
\begin{equation}
\mathbf{e}_{l,k}^l \sim \mathcal{CN} \left( \mathbf{0}, \beta_{l,k}^l \frac{ \sum\limits_{(i,t) \in  \mathcal{P}_{l,k}^{\mathbf{A}} \setminus (l,k)} \beta_{i,t}^l \tilde{p}_{i,t}  +  \sigma^2 }{ \sum\limits_{(i,t) \in \mathcal{P}_{l,k}^{\mathbf{A}}} \beta_{i,t}^l \tilde{p}_{i,t}  +  \sigma^2} \mathbf{I}_M \right).
\end{equation}
\end{corollary}
\begin{proof}
This follows from replacing the terms $\sum_{b=1}^{\tau_p} \hat{p}_{l,k}^b$  and $\sum_{(i,t) \in \mathcal{S} } \beta_{i,t}^l \left( \sum_{b=1}^{\tau_p}  \sqrt{\hat{p}_{i,t}^b \hat{p}_{l,k}^b}  \right)^2$ in Lemma~\ref{lemma: Distribution} by $\tilde{p}_{l,k}$ and $\sum_{(i,t) \in \mathcal{P}_{l,k}^{\mathbf{A}} } \beta_{i,t}^l \tilde{p}_{i,t} \tilde{p}_{l,k}$, and then doing some algebra. 
\end{proof}
Corollary~\ref{co: Distribution} reveals that the quality of the estimated channel heavily depends on both the pilot power control and the pilot reuse set $\mathcal{P}_{l,k}^{\mathbf{A}}$. A proper selection of $\mathcal{P}_{l,k}^{\mathbf{A}}$  mitigates channel estimation errors, and will also reduce the coherent interference during data transmission. Aligned with prior works, in the special case of $\tilde{p}_{l,k} = \tilde{p}, \forall l,k$, the channel estimate and estimation error are obtained for the pilot structure in \eqref{eq: fixedPilotPower}. We now use the distributions in Lemma~\ref{lemma: Distribution} and Corollary~\ref{co: Distribution} to derive lower bounds on the UL ergodic capacity.

\subsection{Uplink Data Transmission}
%\vspace*{-0.1cm}
\setcounter{eqnback}{\value{equation}} \setcounter{equation}{31}
\begin{figure*}
	\begin{equation} \label{eq: SINR_MRC1}
	\mathrm{SINR}_{l,k}^{\mathrm{MR}} = \frac{ M (\beta_{l,k}^l)^2 p_{l,k} \left( \sum\limits_{b=1}^{\tau_p} \hat{p}_{l,k}^b \right)^2 }{  \left( \sum\limits_{(i,t) \in \mathcal{S} } \beta_{i,t}^l \left( \sum\limits_{b=1}^{\tau_p} \sqrt{ \hat{p}_{i,t}^{b}  \hat{p}_{l,k}^{b}} \right)^2  +  \sigma^2 \sum\limits_{b=1}^{\tau_p} \hat{p}_{l,k}^b \right) \left( \sum\limits_{(i,t) \in \mathcal{S} } p_{i,t}  \beta_{i,t}^l + \sigma^2 \right)   +  M \sum\limits_{(i,t) \in \mathcal{S} \setminus (l,k)} p_{i,t}  (\beta_{i,t}^l )^2  \left(\sum\limits_{b=1}^{\tau_p} \sqrt{ \hat{p}_{i,t}^{b}  \hat{p}_{l,k}^{b}} \right)^2  }.
	\end{equation}
\end{figure*}
\setcounter{eqncnt}{\value{equation}}
\setcounter{equation}{\value{eqnback}}

\setcounter{eqnback}{\value{equation}} \setcounter{equation}{34}
\begin{figure*}
	\begin{equation} \label{eq: SINR_MRC2}
	\vtilde[-0.5pt]{\mathrm{SINR}}_{l,k}^{\mathrm{MR}} =	\frac{ M (\beta_{l,k}^l)^2 p_{l,k} \tilde{p}_{l,k}}{  \left( \sum\limits_{(i,t) \in \mathcal{P}_{l,k}^{\mathbf{A}}} \beta_{i,t}^l  \tilde{p}_{i,t}     +  \sigma^2 \right) \left( \sum\limits_{(i,t) \in \mathcal{S}} p_{i,t}  \beta_{i,t}^l + \sigma^2 \right)   +  M \sum\limits_{(i,t) \in \mathcal{P}_{l,k}^{\mathbf{A}} \setminus (l,k) } p_{i,t} \tilde{p}_{i,t}  (\beta_{i,t}^l )^2   }.
	\end{equation}
	\hrulefill
\end{figure*}
\setcounter{eqncnt}{\value{equation}}
\setcounter{equation}{\value{eqnback}}
In the UL data transmission, user~$t$ in cell~$i$ transmits the signal $x_{i,t} \sim \mathcal{CN}(0,1)$. The $M \times 1$ received signal vector at BS~$l$ is the superposition of the transmitted signals
\begin{equation}
\mathbf{y}_l = \sum_{(i,t) \in \mathcal{S} } \sqrt{p_{i,t}} \mathbf{h}_{i,t}^l x_{i,t} + \mathbf{n}_l,
\end{equation}
where $p_{i,t}$ is the transmit power corresponding to the signal $x_{i,t}$ and the additive noise is $\mathbf{n}_l \sim \mathcal{CN} ( \mathbf{0}, \sigma^2 \mathbf{I}_M)$. To detect the transmitted signal, BS~$l$ selects a detection vector $\mathbf{v}_{l,k} \in \mathbb{C}^M$ and applies it to the received signal as
\begin{equation} \label{eq: Signal-Detection}
\mathbf{v}_{l,k}^H \mathbf{y}_l = \sum_{(i,t) \in \mathcal{S} } \sqrt{p_{i,t}}  \mathbf{v}_{l,k}^H \mathbf{h}_{i,t}^l x_{i,t} +  \mathbf{v}_{l,k}^H \mathbf{n}_l.
\end{equation}
A general lower bound on the UL ergodic capacity of user~$k$ in cell~$l$ is computed in \cite{Chien2017a} as
\begin{equation} \label{eq:RateProposedPilot}
R_{l,k} = \left( 1 - \frac{\tau_p}{\tau_c} \right) \log_2 \left(1 + \mathrm{SINR}_{l,k} \right),
\end{equation}
where the effective SINR value, $\mathrm{SINR}_{l,k}$, is
\begin{equation} \label{eq: SINR_k}
 \frac{ p_{l,k} | \mathbb{E} \{ \mathbf{v}_{l,k}^{H} \mathbf{h}_{l,k}^l  \} |^2 }{\sum\limits_{(i,t) \in \mathcal{S} } p_{i,t} \mathbb{E} \{ | \mathbf{v}_{l,k}^{H} \mathbf{h}_{i,t}^{l} |^2 \} - p_{l,k} | \mathbb{E} \{ \mathbf{v}_{l,k}^{H} \mathbf{h}_{l,k}^{l} \} |^2 + \sigma^2 \mathbb{E} \{ \| \mathbf{v}_{l,k} \|^2 \} }.
\end{equation}
The lower bound on the UL ergodic capacity in \eqref{eq:RateProposedPilot} is computed by using the use-and-then-forget bounding technique \cite{Marzetta2016a} and its tightness compared to the other possible bounds is discussed in Appendix D in \cite{Marzetta2016a}. Although the channel capacity for Massive MIMO in the case of imperfect CSI is unknown, we believe that the lower bound in \eqref{eq:RateProposedPilot} is quite close to the actual capacity. This is because the effective noise is comprised of a sum of many uncorrelated terms, it is close to Gaussian. This agrees with the worst-case-is-Gaussian assumption made when obtaining the bound. As a contribution of this paper, we compute a closed form expression for this lower bound in the case of MR detection with
\begin{equation} \label{eq:MRDetection}
 \mathbf{v}_{l,k} = \hat{\mathbf{h}}_{l,k}^l.
\end{equation}
\begin{lemma} \label{Lemma: Achievable_Rate}
This is a highly computationally scalable detection method for Massive MIMO systems.
If the system uses the pilot structure in \eqref{eq: PilotStructure1} and MR detection, the SE in \eqref{eq:RateProposedPilot} for user~$k$ in cell~$l$ becomes
	\begin{equation} \label{eq: RateMRC1}
	R_{l,k}^{\mathrm{MR}} =  \left( 1 - \frac{\tau_p}{\tau_c} \right) \log_2 \left(1 + \mathrm{SINR}_{l,k}^{\mathrm{MR}} \right),
	\end{equation}
	where  $\mathrm{SINR}_{l,k}^{\mathrm{MR}}$ is shown in \eqref{eq: SINR_MRC1}.
\end{lemma}
\begin{proof}
	The proof is available in Appendix \ref{Appendix: Proof-Achievable_Rate}.
\end{proof}
From \eqref{eq: SINR_MRC1}, we notice that it is always advantageous to add more BS antennas since the numerator grows linearly with $M$ (and only some terms in the denominator have the same scaling). The first term in the denominator represents non-coherent interference that only depends on the number of BSs and users, while it is independent of $M$. The second term in the denominator represents coherent interference caused by pilot contamination and it grows linearly with $M$. As a consequence, as $M \to \infty$, we have
\setcounter{eqnback}{\value{equation}} \setcounter{equation}{32}
\begin{equation}
\mathrm{SINR}_{l,k}^{\mathrm{MR}}  \rightarrow \frac{ (\beta_{l,k}^l)^2 p_{l,k} \left( \sum\limits_{b=1}^{\tau_p} \hat{p}_{l,k}^b \right)^2 }{ \sum\limits_{(i,t) \in \mathcal{S} \setminus (l,k)} p_{i,t}  (\beta_{i,t}^l )^2  \left(\sum\limits_{b=1}^{\tau_p} \sqrt{ \hat{p}_{i,t}^b  \hat{p}_{l,k}^b} \right)^2  }.
\end{equation}
This limit depends only on the pilot design (i.e., inner products between pilot signals) and data power. An optimized selection of the power terms $p_{l,k}, \hat{p}_{l,k}^b, \forall l,k,b,$ improves the SE by enhancing the channel estimation quality and reducing the coherent interference.

We also consider the achievable SE for the modified pilot structure in \eqref{eq: PilotStructure3} as shown in Corollary~\ref{Corrolary: Rate1}.
%====================================================================================
%====================================================================================
\begin{corollary} \label{Corrolary: Rate1}
 If the system uses the pilot structure in \eqref{eq: PilotStructure3}, a lower bound on the capacity for user~$k$ in cell~$l$ with uncorrelated Rayleigh fading channels and MR detection is
\begin{equation} \label{eq: Rate2}
R_{l,k}^{\mathrm{MR}} =  \left( 1 - \frac{\tau_p}{\tau_c} \right) \log_2 \left(1 + \vtilde[-0.5pt]{\mathrm{SINR}}_{l,k}^{\mathrm{MR}} \right),
\end{equation}
where the SINR value, $\vtilde[-0.5pt]{\mathrm{SINR}}_{l,k}^{\mathrm{MR}}$, is given in \eqref{eq: SINR_MRC2}.

\end{corollary}
\begin{proof}
This follows as a special case of Lemma~\ref{Lemma: Achievable_Rate}.
\end{proof}
The SE in Corollary~\ref{Corrolary: Rate1} depends explicitly on the choice of $\mathcal{P}_{l,k}^{\mathbf{A}}, \forall l,k,$ thus the optimization of the pilot assignment is a combinatorial problem. We stress that the SINR expressions reflect the joint effects of pilot design, channel estimation quality, pilot contamination, and data power control, in contrast to the MSE that cannot distinguish between pilot contamination and noise. Hence, the SINR is a good metric to consider in the max-min fairness optimization as shown in the next section.

%\vspace*{-0.1cm}
\section{Max-min Fairness Optimization} \label{Section: OptProblem}
%\vspace*{-0.1cm}
In this section, we first utilize the SE expressions in Lemma~\ref{Lemma: Achievable_Rate} and Corollary~\ref{Corrolary: Rate1} to formulate max-min fairness problems with joint pilot and data optimization. We demonstrate that these optimization problems are NP-hard and propose an algorithm to find the globally optimal solution with the pilot design in \eqref{eq: PilotStructure3} by making an exhaustive search over all pilot assignments. In addition, instead of looking for the global optimum, an algorithm to obtain a locally optimal solution in polynomial time is presented when using the new pilot design in \eqref{eq: PilotStructure1}.

\subsection{Problem Formulation} \label{Subsect: Problem}
\setcounter{eqnback}{\value{equation}} \setcounter{equation}{46}
\begin{figure*}
	\begin{equation} \label{eq: SINR_MRCApproximation}
	\doublewidetilde{\mathrm{SINR}}_{l,k}^{\textrm{MR}} = \frac{ M (\beta_{l,k}^l)^2 p_{l,k} \prod\limits_{b=1}^{\tau_p} \left( \hat{p}_{l,k}^b / \alpha_{l,k}^b \right)^{2\alpha_{l,k}^b} }{  \left( \sum\limits_{(i,t) \in \mathcal{S} } \beta_{i,t}^l \left( \sum\limits_{b=1}^{\tau_p} \sqrt{ \hat{p}_{i,t}^{b}  \hat{p}_{l,k}^{b}} \right)^2  +  \sigma^2 \sum\limits_{b=1}^{\tau_p} \hat{p}_{l,k}^b \right) \left( \sum\limits_{(i,t) \in \mathcal{S} } p_{i,t}  \beta_{i,t}^l + \sigma^2 \right)   +  M \sum\limits_{(i,t) \in \mathcal{S} \setminus (l,k)} p_{i,t}  (\beta_{i,t}^l )^2  \left(\sum\limits_{b=1}^{\tau_p} \sqrt{ \hat{p}_{i,t}^{b}  \hat{p}_{l,k}^{b}} \right)^2  }.
	\end{equation}
	\hrulefill
\end{figure*}
\setcounter{eqncnt}{\value{equation}}
\setcounter{equation}{\value{eqnback}}
%\vspace*{-0.3cm}
A key vision of Massive MIMO is to provide uniformly good quality of service for everyone in the network. We will investigate how to optimize the pilots and powers towards this goal. We consider the pilot and data powers as optimization variables. The max-min fairness optimization problem is first formulated for the proposed pilot design in \eqref{eq: PilotStructure1} as\footnote{The optimization problem \eqref{eq: Opt_Prob1} requires coordination among the cells to be solved, but the main target in this paper is to investigate how much the max-min fairness SE can be improved in multi-cell Massive MIMO by joint pilot design and UL power control. One potential way to deal with practical limitations such as backhaul signaling, delays, and scalability is to implement the optimization problem in a distributed manner using dual/primal decomposition  \cite{Bjornson2013d}.}
\setcounter{eqnback}{\value{equation}} \setcounter{equation}{35}
\begin{equation} \label{eq: Opt_Prob1}
\begin{aligned}
&\underset{\{ \hat{p}_{l,k}^b, p_{l,k} \geq 0 \}}{ \mathrm{maximize} } &&  \underset{(l,k)}{\min} \;  \log_2 \left( 1 + \mathrm{SINR}_{l,k}^{\textrm{MR}} \right) \\
& \,\,\text{subject to} && \frac{1}{\tau_p} \sum_{b=1}^{\tau_p} \hat{p}_{l,k}^b \leq P_{\max, l,k}, \forall l,k,\\
&&&  p_{l,k} \leq P_{\max, l,k}^d, \forall l,k,
\end{aligned}
\end{equation}
where $P_{\max, l,k}^d$ is the maximum power that users can provide for each data symbol. Note that this optimization problem jointly generates the pilot signals and performs power control on the pilot and data transmission. The epigraph-form representation of \eqref{eq: Opt_Prob1} is
\begin{subequations} \label{eq: Opt_Prob2}
\begin{alignat}{2}
& \underset{ \xi, \{ \hat{p}_{l,k}^b, p_{l,k} \geq 0 \}}{ \mathrm{maximize} }  && \, \, \, \,\, \, \, \xi \\
&  \text{subject to} &&  \mathrm{SINR}_{l,k}^{\textrm{MR}}  \geq \xi, \forall l,k, \label{P1:a} \\
&&& \frac{1}{\tau_p}\sum_{b=1}^{\tau_p} \hat{p}_{l,k}^b \leq P_{\max, l,k}, \forall l,k, \label{P1:b} \\
&&&  p_{l,k} \leq P_{\max, l,k}^d, \forall l,k. \label{P1:c} 
\end{alignat}
\end{subequations}
From the expression of the SINR constraints in \eqref{P1:a}, we realize that the proposed optimization problem is a signomial program.\footnote{A function $f(x_1, \ldots, x_{N_1}) = \sum_{n=1}^{N_2} c_n \prod_{m= 1}^{N_1} x_m^{a_{n,m}}$ defined in $\mathbb{R}_{+}^{N_1}$ is signomial with $N_2$ terms ($N_2 \geq 2$) if  the exponents $a_{n,m}$ are real numbers and the coefficients $c_n$ are also real but at least one must be negative. In case all $c_n, \forall n,$ are positive, $f(x_1, \ldots, x_{N_1})$ is a posynomial function.} Therefore, the max-min fairness optimization problem is NP-hard in general and seeking the optimal solution has very high complexity in any non-trivial setup \cite{Lange2014a}. However, the power constraints \eqref{P1:b} and \eqref{P1:c} ensure a compact feasible domain and make the SINRs continuous functions of the optimization variables. According to Weierstrass' theorem \cite{Horn2013a}, an optimal solution always exists. 

For the alternative pilot design in \eqref{eq: PilotStructure3}, the max-min fairness optimization problem is formulated as 
\begin{equation} \label{eq: Opt_Prob_Geo}
\begin{aligned}
& \underset{ \xi, \mathbf{A} \in \mathcal{D} , \{ \tilde{p}_{l,k}, p_{l,k} \geq 0 \}}{ \text{maximize} }  && \xi \\
&\,\,\,\, \,\,\,\,\, \text{subject to} &&  \vtilde[-0.5pt]{\mathrm{SINR}}_{l,k}^{\textrm{MR}}  \geq \xi, \forall l,k, \\
&&& \tilde{p}_{l,k} \leq \tau_p P_{\max, l,k}, \forall l,k,\\
&&&  p_{l,k} \leq P_{\max, l,k}^d, \forall l,k.
\end{aligned}
\end{equation}
The optimization problem \eqref{eq: Opt_Prob_Geo} is non-convex since it contains a combinatorial pilot assignment selection. Fortunately the optimal solution to this problem can be obtained by looking up every instance $\mathbf{A}$ in the dictionary $\mathcal{D}$. For each $\mathbf{A}$ we attain the pilot reuse sets $\mathcal{P}_{l,k}^{\mathbf{A}}, \forall l,k$, and then convert \eqref{eq: Opt_Prob_Geo} to a convex problem as shown in Corollary \ref{Co: Geometric-Program}.
\begin{corollary} \label{Co: Geometric-Program}
For a given pilot assignment matrix $\mathbf{A} \in \mathcal{D}$, \eqref{eq: Opt_Prob_Geo} reduces to the geometric program 
\begin{equation} \label{eq: Opt_Prob_Geov1}
\begin{aligned}
& \underset{ \xi, \{ \tilde{p}_{l,k}, p_{l,k} \geq 0 \}}{ \text{maximize} }  && \xi \\
& \,\,\,\,\text{subject to} &&  \vtilde[-0.5pt]{\mathrm{SINR}}_{l,k}^{\mathrm{MR}}  \geq \xi, \forall l,k, \\
&&& \tilde{p}_{l,k} \leq \tau_p P_{\max, l,k}, \forall l,k,\\
&&&  p_{l,k} \leq P_{\max, l,k}^d, \forall l,k.
\end{aligned}
\end{equation}
The optimal solution to \eqref{eq: Opt_Prob_Geov1} is obtained in polynomial time due to its convexity. By checking every instance $\mathbf{A}$ in the dictionary $\mathcal{D}$ and solving the corresponding problem \eqref{eq: Opt_Prob_Geov1}, the global optimum to \eqref{eq: Opt_Prob_Geo} is obtained as the highest objective value to \eqref{eq: Opt_Prob_Geov1}.
\end{corollary}

In more detail, the globally optimal solution to \eqref{eq: Opt_Prob_Geo} is obtained as shown in Algorithm~\ref{Algorithm1}. The $i$th iteration seeks the optimal solution $\xi^{ (i), \mathrm{opt}},\tilde{p}_{l,k}^{ (i), \mathrm{opt}}$, and $p_{l,k}^{ (i), \mathrm{opt}}, \forall l,k$ for given $\{\mathcal{P}_{l,k}^{\mathbf{A}}\}$ by considering \eqref{eq: Opt_Prob_Geov1} as the main cost function. The algorithm is terminated when the iteration index equals $(K!)^{L-1}$. The global optimum to the pilot and data power control together with the pilot reuse set are obtained from the maximum values of all $\{ \xi^{ (i), \mathrm{opt}} \}$. This is a practical issue. We are indeed able to find the solution, but it will take very long time.
 \begin{algorithm}[t]
 	\caption{Global solution to \eqref{eq: Opt_Prob_Geo} by exhaustive search} \label{Algorithm1}
 	\textbf{Input}: Set $i=1$; Select the initial values of $P_{\max,l,k}$ and $P_{\max,l,k}^d$ for $\forall k,l$; Set up the dictionary~$\mathcal{D}$.
 	\begin{itemize}
 		\item[1.] \emph{Iteration} $i$:
 		\begin{itemize}
 			\item[1.1.] Assign the reuse pilot set index $\mathcal{P}_{l,k}^{\mathbf{A}}, \forall l,k,$ by an instance $\mathbf{A} \in \mathcal{D}$.
 			\item[1.2.] Solve the following geometric program to obtain $\xi^{(i),\mathrm{opt}}, p_{l,k}^{ (i),\mathrm{opt}},$ and $\tilde{p}_{l,k}^{ (i),\mathrm{opt}}, \forall l,k :$
 			\begin{equation} \label{Algorithm1:GP}
 			\begin{aligned}
 			& \underset{ \xi^{(i)}, \{ \tilde{p}_{l,k}^{ (i)} , p_{l,k}^{ (i)} \geq 0 \}}{ \mathrm{maximize} }  && \xi^{(i)} \\
 			& \text{subject to} &&  \vtilde[-0.5pt]{\mathrm{SINR}}_{l,k}^{(i), \mathrm{MR}}  \geq \xi^{(i)}, \forall l,k, \\
 			&&&  \tilde{p}_{l,k}^{(i)} \leq \tau_p P_{\max, l,k}, \forall l,k,\\
 			&&&  p_{l,k}^{(i)} \leq  P_{\max, l,k}^d, \forall l,k.
 			\end{aligned}
 			\end{equation}
 		\end{itemize}
 		\item[2.] If $i = (K!)^{L-1}$ $\rightarrow$ Stop. Otherwise, go to Step 3.
 		\item[3.] Restore $\xi^{(i),\mathrm{opt}}$,$\tilde{p}_{l,k}^{(i),\mathrm{opt}}$, and $p_{l,k}^{(i),\mathrm{opt}}$. Set $i = i+1$, then go to Step 1.
 	\end{itemize}
 	\textbf{Output}: Set $i^{\mathrm{opt}} = \underset{i}{\mathrm{argmax}} \{\xi^{(i),\mathrm{opt}} \} $, then the optimal solutions: $\xi^{\mathrm{opt}} = \xi^{ (i^{\mathrm{opt}}),\mathrm{opt}} $, $\tilde{p}_{l,k}^{\mathrm{opt}} = \tilde{p}_{l,k}^{(i^{\mathrm{opt}} ),\mathrm{opt}}$, and  $p_{l,k}^{\mathrm{opt}} = p_{l,k}^{(i^{\mathrm{opt}} ),\mathrm{opt}}, \forall l,k.$
 \end{algorithm}
 Algorithm~\ref{Algorithm1} is computationally heavy since the number of iterations grows rapidly with $K$ and $L$,  but it obtains the global optimum to the max-min SE problem \eqref{eq: Opt_Prob_Geo}. Specifically, the main cost of each iteration in Algorithm~\ref{Algorithm1} is the geometric program \eqref{Algorithm1:GP} which includes $2KL + 1$ optimization variables and $3KL$ constraints. Based on \cite{Boyd2004a}, in general, the computational complexity of this algorithm is of the order of 
 \begin{equation} \label{eq:Complexity1}
 \mathcal{O}\left( (K!)^{L-1} \max\{ (2KL+1)^3, 3KL (2KL+1)^2, F_1 \}\right ), 
 \end{equation}
 where $F_1$ is the cost of evaluating the first and second derivatives of the objective and constraint functions in \eqref{Algorithm1:GP}. Therefore, this approach will serve as a benchmark for comparison in Section \ref{Section: Experimental Result}. For the sake of completeness, we also include another benchmark whereas the data powers are fixed at their maximum value then Algorithm~\ref{Algorithm1} is solved with respect to the remaining pilot power variables, as was done in  our previous work \cite{Chien2017b}.  
 
\subsection{Local Optimality Algorithm}

%\vspace*{-0.3cm}
This subsection provides a method to obtain a local optimum to the optimization problem \eqref{eq: Opt_Prob2}. To this end, the signomial SINR constraints are converted to monomial ones by using the weighted arithmetic mean-geometric mean inequality \cite{Chiang2007b} stated in Lemma~\ref{Lemma: Local_Approximation}.\footnote{ A function $f(x_1, \ldots, x_{N_1}) = c\prod_{m=1}^{N_1} x_m^{a_m}$ defined in $\mathbb{R}_{+}^{N_1}$ is monomial if the coefficient $c >0$ and the exponents $a_m, \forall m,$ are real numbers. }
\begin{lemma} \cite[Lemma~1]{Chiang2007b}  \label{Lemma: Local_Approximation}
Assume that a posynomial function $g(x)$ is defined from the set of $\tau_p$ monomials  $\{ u_1 (x), \ldots, u_{\tau_p} (x) \}$ as
\begin{equation}
g(x) = \sum_{b=1}^{\tau_p} u_b (x),
\end{equation}
then it is lower bounded by a monomial function $\tilde{g}(x)$ as
\begin{equation}
g(x) \geq \tilde{g}(x) = \prod_{b=1}^{\tau_p} \left(  u_{b}(x) / \alpha_b \right)^{\alpha_b},
\end{equation}
where $\alpha_b$ is a non-negative weight corresponding to $u_{b} (x)$. We say that $\tilde{g}(x_0)$ is the best approximation to $g(x_0)$ near the point $x_0$ in the sense of the first order Taylor expansion, if the weight $\alpha_b $  is selected as
\begin{equation} \label{eq: WeightDef}
\alpha_b = \frac{u_b(x_0)}{\sum_{b=1}^{\tau_p} u_b (x_0)} .
\end{equation}
\end{lemma}
By using this lemma, the max-min fairness optimization problem \eqref{eq: Opt_Prob2} is converted to a geometric program by bounding the term $\sum_{b=1}^{\tau_p}  \hat{p}_{l,k}^b$ in the numerators of the SINR constraints:
\begin{equation} \label{eq_: Power_Approximation}
 \sum_{b=1}^{\tau_p}  \hat{p}_{l,k}^b \geq \prod_{b=1}^{\tau_p} \left( \hat{p}_{l,k}^b / \alpha_{l,k}^b \right)^{\alpha_{l,k}^b},
\end{equation}
where $\alpha_{l,k}^b$ is the weight value corresponding to $\hat{p}_{l,k}^b$. This leads to a lower bound on the SINR value for user~$k$ in cell~$l$ obtained as
\begin{equation} \label{eq: SINRBound}
 \mathrm{SINR}_{l,k}^{\textrm{MR}} \geq  \doublewidetilde{\mathrm{SINR}}_{l,k}^{\textrm{MR}},
\end{equation}
where the $\doublewidetilde{\mathrm{SINR}}_{l,k}^{\textrm{MR}}$ value is presented in \eqref{eq: SINR_MRCApproximation}.

The optimal solution $\xi$ to the max-min SE optimization problem \eqref{eq: Opt_Prob2} is lower bounded by solving the geometric program
\setcounter{eqnback}{\value{equation}} \setcounter{equation}{47}
\begin{equation} \label{eq: Opt_Prob3}
\begin{aligned}
& \underset{ \xi, \{ \hat{p}_{l,k}^b, p_{l,k} \geq 0 \}}{ \mathrm{maximize} }  && \xi \\
& \text{subject to} &&  \doublewidetilde{\mathrm{SINR}}_{l,k}^{\textrm{MR}}  \geq \xi, \forall l,k, \\
&&& \frac{1}{\tau_p} \sum_{b=1}^{\tau_p} \hat{p}_{l,k}^b \leq P_{\max, l,k}, \forall l,k,\\
&&&  p_{l,k} \leq P_{\max, l,k}^d, \forall l,k.\\
\end{aligned}
\end{equation}
By virtue of the successive approximation technique \cite{Marques1978a}, a locally optimal  Karush-Kuhn-Tucker (KKT) point to the max-min fairness optimization problem \eqref{eq: Opt_Prob2} can be obtained  if we solve \eqref{eq: Opt_Prob3} iteratively as shown in Theorem \ref{Theorem: KKTpoint}.
\begin{theorem} \label{Theorem: KKTpoint}
Selecting an initial point $\hat{p}_{l,k}^{b, (0)}, \forall l,k,b,$ in the feasible domain and solving \eqref{eq: Opt_Prob3} in an iterative manner by consecutively updating the weight values $\alpha_{l,k}^b$ from the optimal powers of the previous iteration, the solution will converge to a KKT local point to  \eqref{eq: Opt_Prob2}.
\end{theorem}
\begin{proof}
The proof is adapted from the general framework in \cite{Marques1978a} and is sketched in Appendix~\ref{Appendix: Proof-Theorem:KKTPoint}.
\end{proof}
In particular, we first select the initial powers $\hat{p}_{l,k}^{b,(0)}, \forall l,k,b$ that satisfy $\hat{p}_{l,k}^{b,(0)} \geq 0$, $\sum_{b=1}^{\tau_p} \hat{p}_{l,k}^{b,(0)} \leq \tau_p P_{\max, l,k}$ . Then the corresponding weight values are computed as in \eqref{eq: WeightDef}. Furthermore, in each iteration, the SINR constraints are converted to the corresponding monomials by bounding the pilot power of user~$k$ in cell~$l$ as in \eqref{eq: SINR_MRCApproximation}, by using the weight values computed from the optimal pilot powers in the previous iteration. The pilot and data allocation solution is obtained by solving the geometric program \eqref{eq: Opt_Prob3} before the weight values are updated again at the end of each iteration. We repeat the procedure until this algorithm has converged to a KKT point. The convergence can be declared, for example, when the variation between two consecutive iterations is sufficient small. The proposed algorithm for obtaining a locally optimal solution is summarized in Algorithm~\ref{Algorithm2}. Note that one can also fix the data powers and only optimize the pilot signals in Algorithm~\ref{Algorithm2}, as was done in our previous work \cite{Chien2017b}. Algorithm~\ref{Algorithm2} involves optimization with $KL(\tau_p +1) +1$ variables and $3KL$ constraints, and it has a computational complexity of the order of~\cite{Boyd2004a} \footnote{The exact complexity or the runtime of the proposed algorithms are not suitable metrics since they depend significantly on the computer configuration and how much time is spent to optimize the implementations. However \eqref{eq:Complexity1} and \eqref{eq:Complexity2} give basic insights into the general computational complexity scaling.}
	\begin{equation} \label{eq:Complexity2}
	\mathcal{O}\left( N \max\{ (KL(\tau_p + 1) + 1)^3, 3KL (KL(\tau_p + 1) + 1)^2, F_2 \}\right ),
	\end{equation}
where $F_2$ is the cost of evaluating the first and second derivatives of the objective and constraint functions in \eqref{eq: Opt_Prob3}. $N$ is the number of iterations needed for this algorithm to converge to the KKT point. Even though each iteration in Algorithm~\ref{Algorithm2} is more costly than in Algorithm~\ref{Algorithm1} since we carefully design powers for all pilot signals, the successive approximation approach converges after only a few iterations. 
\begin{algorithm}[t]
\caption{Successive approximation algorithm for \eqref{eq: Opt_Prob2}} \label{Algorithm2}
\textbf{Input}: Set $i=1$; Select the maximum powers $P_{\max, l,k}$ and $P_{\max, l,k}^d$  for $\forall l,k;$ Select the initial values of powers $\hat{p}_{l,k}^{b,(0)}$ for $\forall l,k,b$; Compute the weight values $\alpha_{l,k}^{b,(1)} = \hat{p}_{l,k}^{b,(0)} /\sum_{b=1}^{\tau_p} \hat{p}_{l,k}^{b,(0)}, \forall l,k,b.$
\begin{itemize}
  \item[1.] \emph{Iteration} $i$:
  \begin{itemize}
 \item[1.1.] Solve the geometric program \eqref{eq: Opt_Prob3} with $\alpha_{l,k}^b = \alpha_{l,k}^{b,(i)}$ to get the optimal values $\xi^{(i), \mathrm{opt}}$, $\hat{p}_{l,k}^{b, (i), \mathrm{opt}}, \forall l,k,b,$ and $p_{l,k}^{(i), \mathrm{opt}}, \forall l,k.$
 \item[1.2.] Update the weight values: $\alpha_{l,k}^{b,(i+1)} = \hat{p}_{l,k}^{b,(i) ,\mathrm{opt}} /\sum_{b=1}^{\tau_p} \hat{p}_{l,k}^{b,(i), \mathrm{opt}}, \forall l,k,b.$
\end{itemize}
 \item[2.] If Stopping criterion satisfied  $\rightarrow$ Stop. Otherwise, go to Step 3.
 \item[3.] Set $\xi^{\mathrm{opt}} = \xi^{(i),\mathrm{opt}}$, $\hat{p}_{l,k}^{b,\mathrm{opt}} = \hat{p}_{l,k}^{b,(i),\mathrm{opt}}, \forall l,k,b$, and  $p_{l,k}^{\mathrm{opt}} = p_{l,k}^{(i),\mathrm{opt}}, \forall l,k$; Set $i = i+1$, go
   to Step 1.
\end{itemize}
\textbf{Output}: The solutions $\xi^{\mathrm{opt}}, \hat{p}_{l,k}^{b,\mathrm{opt}}, \forall l,k,b, $ and $p_{l,k}^{\mathrm{opt}}, \forall l,k.$
\end{algorithm}

\section{Pilot Optimization for Cellular Massive MIMO Systems with Hardware Impairments} \label{section: Hardware_Impairments}
%\vspace*{-0.3cm}
The previous sections considered Massive MIMO with ideal transceiver hardware. However practical transceivers are non-ideal in the sense of creating distortions that can have a substantial impact on the SE. Thanks to the non-coherent combining of the independent distortion caused at each of the BS antennas, the hardware impairments at the BS can be neglected in Massive MIMO \cite{Bjornson2014a}. However, the distortion caused by the single-antenna users leads to coherent self-interference, which can be viewed as pilot contamination that the user causes to itself. The pilot optimization problem is fundamentally different when accounting for self-interference, because pilot contamination from distant users can now be neglected when it is substantially weaker than the self-interference. In this section, we investigate how the distortions from hardware impairments at the users affect the proposed pilot design and the optimization problems.

\subsection{Channel Estimation under Hardware Impairments}
%\vspace*{-0.3 cm}
We model the distortion caused by the hardware impairments at a user as a reduction of the signal amplitude by $\sqrt{1- \epsilon^2}$ and the addition of Gaussian distortion with a power that equals the reduction in signal power \cite{Bjornson2014a}. We refer to $\epsilon$ as the impairment level. The received pilot signal at BS~$l$ is 
\begin{equation} \label{eq:received-pilot-impairments}
\mathbf{Y}_l  = \sum_{(i,t) \in \mathcal{S}} \mathbf{h}_{i,t}^{l} \left( \sqrt{1 - \epsilon^2} \pmb{\psi}_{i,t}^{H} + \pmb{\varepsilon}_{p,i,t}^H \right) +  \mathbf{N}_l.
\end{equation}
 Similar to \cite{Zhang2012a}, we assume that the UL distortion term of user~$t$ in cell~$i$ is distributed as
\begin{equation}
 \pmb{\varepsilon}_{p,i,t} \sim \mathcal{CN} \left(\mathbf{0}, \pmb{\Lambda}_{i,t}  \right),
\end{equation}
where $\pmb{\Lambda}_{i,t} = \epsilon^2 \mathrm{diag} ( \hat{p}_{i,t}^1, \ldots, \hat{p}_{i,t}^{\tau_p})$. Since the distortion term is multiplied with the unknown channel, just as the pilot signals, the channel estimation is more complicated in this case. An estimate $\hat{\mathbf{h}}_{l,k}$ of $\mathbf{h}_{l,k}$ channel can be obtained from
\begin{equation}
\begin{split}
& \mathbf{y}_{l,k} = \mathbf{Y}_{l} \pmb{\psi}_{l,k}  = \\
& \sum_{(i,t) \in \mathcal{S}} \mathbf{h}_{i,t}^{l} \left( \sqrt{1 - \epsilon^2 } \pmb{\psi}_{i,t}^H \pmb{\psi}_{l,k} + \pmb{\varepsilon}_{p,i,t}^H \pmb{\psi}_{l,k} \right) + \mathbf{N}_{l} \pmb{\psi}_{l,k},
\end{split}
\end{equation}
where the received signal in \eqref{eq:received-pilot-impairments} is correlated with the pilot sequence used by the user of interest.
We note that the MMSE estimator is intractable, but we can derive the linear minimum mean square error (LMMSE) estimator as shown in Lemma~\ref{Lemma: ChannelEstimate}.
\begin{lemma} \label{Lemma: ChannelEstimate}
Under hardware impairments, the LMMSE channel estimate of $\mathbf{h}_{l,k}^l$ at BS~$l$ is
\begin{equation}
\hat{\mathbf{h}}_{l,k}^l = \frac{\sqrt{1 - \epsilon^2 } \beta_{l,k}^l \sum\limits_{b=1}^{\tau_p} \hat{p}_{l,k}^b }{ \sum\limits_{(i,t) \in \mathcal{S}} \beta_{i,t}^l \kappa_{i,t} + \sigma^2 \sum\limits_{b=1}^{\tau_p} \hat{p}_{l,k}^b } \mathbf{y}_{l,k},
\end{equation}
where 
\begin{equation} \label{eq: Kappa_it}
\kappa_{i,t} = (1 - \epsilon^2)\left( \sum_{b=1}^{\tau_p}  \sqrt{\hat{p}_{i,t}^b \hat{p}_{l,k}^b }\right)^2 + \epsilon^2\sum_{b=1}^{\tau_p}   \hat{p}_{i,t}^b \hat{p}_{l,k}^{b}.
\end{equation}
The channel estimate $\hat{\mathbf{h}}_{l,k}^l$ and estimation error $\pmb{e}_{l,k}^{l}$ are uncorrelated, but not independent, have zero mean, and the covariance matrices as
\begin{equation}
\mathrm{Cov} \{ \hat{\mathbf{h}}_{l,k}^l, \hat{\mathbf{h}}_{l,k}^l \} = \frac{(1 - \epsilon^2) (\beta_{l,k}^l)^2 \left(\sum\limits_{b=1}^{\tau_p} \hat{p}_{l,k}^b \right)^2}{\sum\limits_{(i,t) \in \mathcal{S}} \beta_{i,t}^l \kappa_{i,t} + \sigma^2 \sum\limits_{b=1}^{\tau_p} \hat{p}_{l,k}^b } \mathbf{I}_M ,
\end{equation}
\begin{equation}
\mathrm{Cov} \{ \mathbf{e}_{l,k}^l, \mathbf{e}_{l,k}^l\} = \frac{\sum\limits_{(i,t) \in \mathcal{S} \setminus (l,k)} \beta_{i,t}^l \kappa_{i,t} + \sigma^2 \sum\limits_{b=1}^{\tau_p} \hat{p}_{l,k}^b + \epsilon^2\sum\limits_{b=1}^{\tau_p} (\hat{p}_{l,k}^{b} )^2}{\sum\limits_{(i,t) \in \mathcal{S}} \beta_{i,t}^l \kappa_{i,t} + \sigma^2 \sum\limits_{b=1}^{\tau_p} \hat{p}_{l,k}^b } \mathbf{I}_M.
\end{equation}
\end{lemma}
\begin{proof}
The proof follows the standard LMMSE estimation technique as described in \cite{Kay1993a}.
\end{proof}
This lemma shows that the variance of the estimation error grows with the impairment level. In the special case of $\epsilon = 0$, we obtain ideal hardware and the estimate is the same as in Lemma~\ref{lemma: Distribution}. In contrast, the estimate is zero (equal to the mean value) if $\epsilon = 1$. We now use the statistics in Lemma~\ref{Lemma: ChannelEstimate} to derive a lower bound on the UL achievable SE and formulate the corresponding optimization problems.

\subsection{UL Data Transmission and Max-min Fairness Optimization under Hardware Impairments}
%\vspace*{-0.3 cm}
Similar to \cite{Bjornson2014a}, the received signal at BS~$l$ during data transmission for the case of hardware impairments at the users is modeled as
\begin{equation}
\mathbf{y}_l = \sum_{(i,t) \in \mathcal{S}} \mathbf{h}_{i,t}^l \left( \sqrt{(1 - \epsilon^2) p_{i,t}} x_{i,t} + \varepsilon_{i,t} \right) + \mathbf{n}_l.
\end{equation}
We assume that the distortion caused by user~$t$ in cell~$i$ is worst-case Gaussian distributed with $\varepsilon_{i,t} \sim \mathcal{CN}(0, \epsilon^2  p_{i,t} )$. A lower bound on the UL ergodic capacity of user~$k$ in cell~$l$ is obtained in Theorem \ref{Theorem: ImperfectHWRate}.
\begin{theorem} \label{Theorem: ImperfectHWRate}
Under hardware impairments, if the system uses the pilot structure in \eqref{eq: PilotStructure1} and MR detection, the SE in \eqref{eq:RateProposedPilot} for user~$k$ in cell~$l$ becomes
\begin{equation} \label{Rate:Hardware}
R_{l,k} = \left( 1 - \frac{\tau_p}{\tau_c} \right) \log_2 \left(1 + \overline{\mathrm{SINR}}_{l,k}^{\mathrm{MR}} \right),
\end{equation}
where the effective SINR value, $\overline{\mathrm{SINR}}_{l,k}^{\mathrm{MR}}$, is 
\begin{equation} \label{eq: SINR_MRCimpairment}
\frac{M(1-\epsilon^2)^2 p_{l,k} (\beta_{l,k}^l)^2 \left(\sum\limits_{b=1}^{\tau_p} \hat{p}_{l,k}^b \right)^2 }{ \left( \sum\limits_{(i,t) \in \mathcal{S}}  \beta_{i,t}^l \kappa_{i,t}  + \sigma^2 \sum\limits_{b=1}^{\tau_p} \hat{p}_{l,k}^b \right)\left( \sum\limits_{(i,t) \in \mathcal{S}} p_{i,t} \beta_{i,t}^l + \sigma^2 \right) + \eta_{l,k} },
\end{equation}
where
\begin{equation}
\begin{split}
 \eta_{l,k}  =& M \sum_{(i,t) \in \mathcal{S} \setminus (l,k)} \kappa_{i,t} p_{i,t} (\beta_{i,t}^l)^2 + M \epsilon^2 p_{l,k} (\beta_{l,k}^l)^2  \sum_{b=1}^{\tau_p} (\hat{p}_{l,k}^b)^2\\
 & + M \epsilon^2 (1- \epsilon^2)  \left( \sum_{b=1}^{\tau_p} \hat{p}_{l,k}^b \right)^2 p_{l,k} (\beta_{l,k}^l)^2.
 \end{split}
\end{equation}
\end{theorem}
\begin{proof}
The proof is given in Appendix \ref{Appendix: Proof-Theorem:ImperfectHWRate}.
\end{proof}
In comparison to having ideal hardware, the hardware impairments reduce the coherent gain in the numerator of the SINR by a factor $(1- \epsilon^2)^2$. There is now coherent self-interference from user~$k$, which behave similarly to pilot contamination. This contamination can be relatively large since the user generally has a stronger channel than the pilot-contaminating interferers. In the asymptotic limit when $M \to \infty$, only the signal term in the numerator and the term $ \eta_{l,k}$ with self-interference and coherent interference from pilot contamination remain.

 Based on the SE expression with hardware impairments, we consider the max-min fairness optimization problem
\begin{equation} \label{eq: Opt_ProbImparments}
\begin{aligned}
& \underset{ \xi, \{p_{l,k}, \hat{p}_{l,k}^b \geq 0 \}}{ \mathrm{maximize} }  && \xi \\
& \text{subject to} &&  \overline{\mathrm{SINR}}_{l,k}^{\textrm{MR}}  \geq \xi, \forall l,k, \\
&&& \frac{1}{\tau_p} \sum_{b=1}^{\tau_p} \hat{p}_{l,k}^b \leq P_{\max, l,k}, \forall l,k,\\
&&&  p_{l,k} \leq P_{\max, l,k}^d, \forall l,k.\\
\end{aligned}
\end{equation}
We stress that \eqref{eq: Opt_ProbImparments} is a generalization of \eqref{eq: Opt_Prob2}, but all algorithms we proposed for the case of ideal hardware can be readily extended. Similar to \eqref{eq: Opt_Prob2}, by utilizing the successive approximation method similar to Algorithm~\ref{Algorithm2}, a local optimum to \eqref{eq: Opt_ProbImparments} is obtained in polynomial time. The performance with heuristic pilot designs and power allocation is obtained directly from Theorem 2 by using the corresponding values on $\hat{p}^b_{l,k}$ and for example, in the case of using the combinatorial pilot structure in \eqref{eq: PilotStructure3}, similar to the procedures in Algorithm~\ref{Algorithm1}, the global max-min SE solution under hardware impairments is obtained by jointly optimizing pilot and data powers with exhaustive search over all $\mathbf{A} \in \mathcal{D}$.

\section{Generalization to Correlated Rayleigh fading} \label{Section:CorrUncorr}
Since the propagation channels may be spatially correlated in practice, we now consider a correlation model where the channel between user $t$ in cell $i$ and BS $l$ is modeled as
\begin{equation}
\mathbf{h}_{i,t}^l \sim \mathcal{CN} \left( \mathbf{0}, \mathbf{R}_{i,t}^l \right),
\end{equation}
where $\mathbf{R}_{i,t}^l \in \mathbb{C}^{M \times M}$ is the covariance matrix with equal diagonal elements denoted by $\beta_{i,t}^l$ since the BS antennas are co-located. This assumption also leads to convex optimization problems. Meanwhile the non-zero off-diagonal elements represent the spatial correlation. By using the pilot transmission model in \eqref{eq: Received_Pilot} and element-wise MMSE estimation \cite{Victor2017a, Biguesh2006a}, the channel estimate of $\mathbf{h}_{l,k}^l$ is
\begin{equation}
	\hat{ \mathbf{h} }_{l,k}^l= \varrho_{l,k}^l \mathbf{y}_{l,k},
\end{equation} 
where
\begin{equation} \label{eq: varrho_blk}
\varrho_{l,k}^l = \frac{\beta_{l,k}^l \sum\limits_{b=1}^{\tau_p} \hat{p}_{l,k}^b }{ \sum\limits_{(i,t) \in \mathcal{S}} \beta_{i,t}^l \left( \sum\limits_{b=1}^{\tau_p}  \sqrt{\hat{p}_{i,t}^b \hat{p}_{l,k}^b}  \right)^2 + \sigma^2  \sum\limits_{b=1}^{\tau_p} \hat{p}_{l,k}^b  }.
\end{equation}
After that, a closed-form expression of the achievable SE is obtained by computing the moments of Gaussian distributions in \eqref{eq: SINR_k}:
\begin{equation} \label{eq: Rat_Corr}
R_{l,k}^{\mathrm{MR}} =  \left( 1 - \frac{\tau_p}{\tau_c} \right) \log_2 \left(1 + \widehat{\mathrm{SINR}}_{l,k}^{\mathrm{MR}} \right), 
\end{equation} 
where the effective SINR, denoted by $\widehat{\mathrm{SINR}}_{l,k}^{\mathrm{MR}}$, is 
\begin{equation} \label{eq:CorrelatedSINR}
\widehat{\mathrm{SINR}}_{l,k}^{\mathrm{MR}} = \frac{M (\beta_{l,k}^l)^2 p_{l,k} \left( \sum_{b=1}^{\tau_p} \hat{p}_{l,k}^b \right)^2}{ \mathcal{I}_{\textrm{non}} + \mathcal{I}_{\textrm{coh}}}.
\end{equation}
By denoting the trace of a matrix as $\mathrm{tr}(\cdot)$, the non-coherent interference and thermal noise $\mathcal{I}_{\textrm{non}}$, and the coherent interference $\mathcal{I}_{\textrm{coh}}$ in \eqref{eq:CorrelatedSINR} are respectively given as
\begin{align*}
\mathcal{I}_{\textrm{non}}&= \sum\limits_{(i,t)  \in \mathcal{S}}\frac{p_{i,t}}{M}\sum\limits_{(i',t')  \in \mathcal{S}} \mathrm{tr}\left( \mathbf{R}_{i,t}^l \mathbf{R}_{i',t'}^l  \right) \left( \sum_{b=1}^{\tau_p} \sqrt{\hat{p}_{i',t'}^b \hat{p}_{l,k}^b}  \right)^2\\
& \; \; \; +  \sigma^2 \sum_{b=1}^{\tau_p} \hat{p}_{l,k}^b  \sum\limits_{(i,t)  \in \mathcal{S}}p_{i,t}\beta_{i,t}^l
 + \frac{\beta_{l,k}^l \sigma^2 \sum_{b=1}^{\tau_p} \hat{p}_{l,k}^b}{\varrho_{l,k}^l} \numberthis , \\
\mathcal{I}_{\textrm{coh}} &= M\sum\limits_{(i,t) \in \mathcal{S}\setminus (l,k)} p_{i,t} (\beta_{i,t}^l)^2   \left( \sum_{b=1}^{\tau_p} \sqrt{\hat{p}_{i,t}^b \hat{p}_{l,k}^b}  \right)^2. \numberthis
\end{align*}
From the SINR expression in \eqref{eq:CorrelatedSINR}, the correlation between channels only effects the non-coherent interference while the coherent interference remains the same as in the case of uncorrelated Rayleigh fading. We now can formulate the max-min fairness optimization problem as
\begin{equation} \label{eq: Opt_CorrProb}
\begin{aligned}
& \underset{ \xi, \{p_{l,k}, \hat{p}_{l,k}^b \geq 0 \}}{ \mathrm{maximize} }  && \xi \\
& \text{subject to} &&  \widehat{\mathrm{SINR}}_{l,k}^{\textrm{MR}}  \geq \xi, \forall l,k, \\
&&& \frac{1}{\tau_p} \sum_{b=1}^{\tau_p} \hat{p}_{l,k}^b \leq P_{\max, l,k}, \forall l,k,\\
&&&  p_{l,k} \leq P_{\max, l,k}^d, \forall l,k.\\
\end{aligned}
\end{equation}
The optimization problem \eqref{eq: Opt_CorrProb} has the same general structure as the problems considered in Section~\ref{Section: OptProblem}. For example, with the proposed pilot design, we can directly use Algorithm~\ref{Algorithm2} to obtain a local optimum, while the global solution of the max-min fairness optimization can obtained with the pilot design in \eqref{eq: PilotStructure3} by making simple modifications to Algorithm~\ref{Algorithm1}. 
%=====================================================
%=====================================================
\section{Numerical Results} \label{Section: Experimental Result}
%\vspace*{-0.3cm}
In this section, we use numerical simulations to quantify and discuss the effectiveness of the proposed pilot designs, using the exact closed-form expressions of the SE in \eqref{eq: RateMRC1}, \eqref{eq: Rate2}, \eqref{Rate:Hardware}, and \eqref{eq: Rat_Corr}. A Massive MIMO system with a coverage area $1$ km$^2$ comprising of four square cells is considered. In each cell, a BS is located at the center, while the $K$ users are uniformly distributed at distance not closer to the BS than $35$ m. To even out interference, the coverage area is wrapped around, and therefore one BS has eight neighbors.  The coherence interval contains $200$ symbols and the system bandwidth is $20$ MHz. The noise variance is $-96$ dBm. We deploy the 3GPP LTE model from \cite{LTE2010b} where the large-scale fading coefficient $\beta_{i,t}^l $ [dB] is
\begin{equation}
\beta_{i,t}^l = -148.1 - 37.6 \log_{10} d_{i,t}^l + z_{i,t}^l,
\end{equation}
where, $d_{i,t}^l$ denotes the distance in km between user~$t$ in cell~$i$ and BS~$l$. The shadow fading $z_{i,t}^l$ has a Gaussian distribution with zero mean and the standard derivation $7$~dB.\footnote{Shadow fading realizations were sometimes regenerated to ensure that the home BS has the largest large-scale fading to its users (i.e., $\beta_{l,k}^l$ is the maximum over all $\beta_{i,k}^l, i = 1, \ldots, L$). This makes sure that the coverage area of each BS is a square also with shadow fading, while retaining the macro-diversity towards shadow fading that exists in practice.} The maximum pilot and data power constraints are  $P_{\max,l,k}= P_{\max, l,k}^d  = 200$ mW, $\forall l,k$.

For Algorithm~\ref{Algorithm2} and its modification with fixed data power (which was considered in \cite{Chien2017b}), we observe better performance with a hierarchical initialization of $p_{l,k}^{b,(0)}$ than with an all-equal initialization. Consequently, we initialize $p_{l,k}^{b,(0)}$ as uniformly distributed over the range $[0; P_{\max,l,k}]$.\footnote{The maximum power settings indicate relatively low median SNRs at the cell-edge users of the network.} Algorithm~\ref{Algorithm2} converges quite fast, so the stopping criteria was specified in number of iterations (e.g., $15$ iterations). The proposed algorithms are compared with related works and exhaustive search:
\begin{figure}[t]
	\centering
	\includegraphics[trim=1.3cm 8cm 1.8cm 8.5cm, clip=true, width=3.3in]{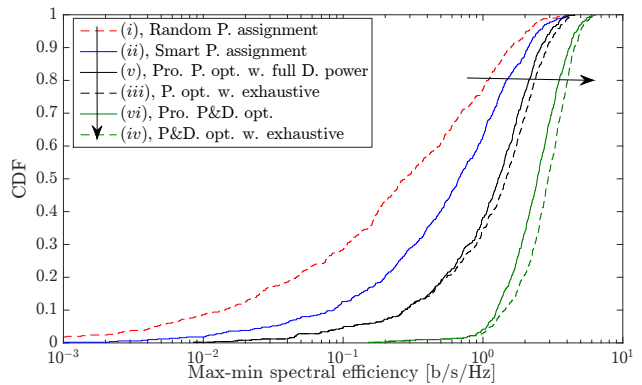} %\vspace*{-0.2cm}
	\caption{ Cumulative distribution function (CDF) versus the max-min SE [b/s/Hz] with $M = 300, K = \tau_p = 2$.}
	\label{Fig-CDF-2K2B}
	%\vspace*{-0.2cm}
\end{figure}
\begin{figure}[t]
	\centering
	\includegraphics[trim=1.3cm 8cm 1.8cm 8.5cm, clip=true, width=3.3in]{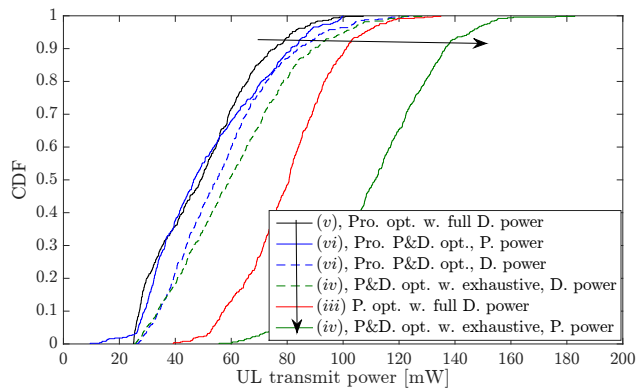} %\vspace*{-0.2cm}
	\caption{Cumulative distribution function (CDF) versus the UL transmit power [mW] with $M= 300,  K = \tau_p =2$.}
	\label{Fig-CDF-Power-2K2B}
	%\vspace*{-0.2cm}
\end{figure}

\begin{itemize}
\item[\textit{(i)}] \textit{Random pilot assignment}, as considered in \cite{Jose2011b,Bjornson2016a}. The same pilots are reused in every cell and assigned randomly to the users within the cell. Equal pilot and data powers of $ 200$ mW are used by all users. This method is denoted as \textit{Random P. assignment} in the figures.
\item[\textit{(ii)}] \textit{Smart pilot assignment}, as proposed in \cite{Xu2015a}. Orthogonal pilots are assumed within a cell and reused in every cell. They are assigned to the users based on the mutual interference information, determined by the large-scale fading coefficients. Equal pilot and data powers $ 200$ mW are used by all users. This method is denoted as \textit{Smart P. assignment} in the figures.
\item[\textit{(iii)}] \textit{Pilot power optimization with exhaustive search}, as proposed in \cite{Chien2017b}, utilizes the pilot structure in \eqref{eq: PilotStructure3} and the optimal solution to the max-min SE is obtained by a modification of Algorithm~\ref{Algorithm1} with fixed data powers of $p_{l,k} = 200$ mW at all users. This method is denoted as \textit{P. opt. w. exhaustive} in the figures.
\item[\textit{(iv)}] \textit{Joint pilot and data power optimization with exhaustive search} utilizes the pilot structure in \eqref{eq: PilotStructure3} and the optimal solution to the max-min SE is obtained as shown in Algorithm~\ref{Algorithm1}. This method is denoted as \textit{P\&D opt. w. exhaustive} in the figures.
\item[\textit{(v)}] \textit{Proposed pilot optimization with full data power}, as presented in \cite{Chien2017b},  utilizes the pilot structure in \eqref{eq: PilotStructure1} and the local optimum to the max-min SE is obtained by a modification of Algorithm~\ref{Algorithm2} with fixed data powers of $p_{l,k} = 200$ mW at all users. This method is denoted as \textit{Pro. P. opt. w. full D. power} in the figures.
\item[\textit{(vi)}] \textit{Proposed joint pilot and data power optimization} utilizes the pilot structure in \eqref{eq: PilotStructure1} and the local optimum to the max-min SE is obtained as shown in Algorithm~\ref{Algorithm2}. This method is denoted as \textit{Pro. P\&D opt.} in the figures.
\end{itemize}
The SE is measured over different random user locations and shadow fading realizations. The SE achieved by $(i)$--$(iv)$ are also averaged over different pilot reuse locations, while for $(v)$ and $(vi)$ the SE is also averaged over different initializations of $\hat{p}_{l,k}^{b, (0)},\forall l,k,b$. The solutions to the optimization problems are obtained by utilizing the MOSEK solver \cite{Mosek} in CVX \cite{cvx2015}.
\begin{figure}[t]
	\centering
	\includegraphics[trim=1.3cm 8cm 1.8cm 8.5cm, clip=true, width=3.3in]{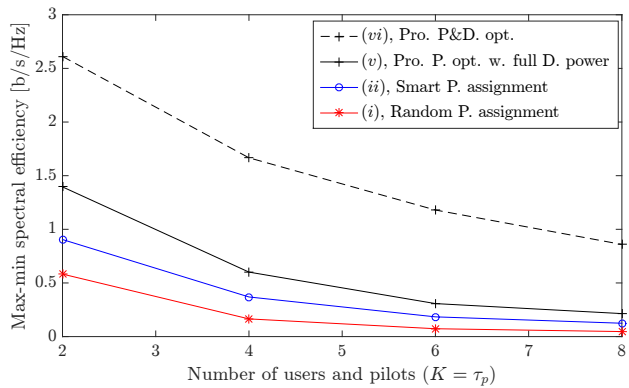} %\vspace*{-0.2cm}
	\caption{ Max-min SE [b/s/Hz] versus the number of users and pilots, $M = 300$.}
	\label{Fig-VariousK}
	%\vspace*{-0.2cm}
\end{figure}
\begin{figure}[t]
	\centering
	\includegraphics[trim=1.3cm 8cm 1.8cm 8.5cm, clip=true, width=3.3in]{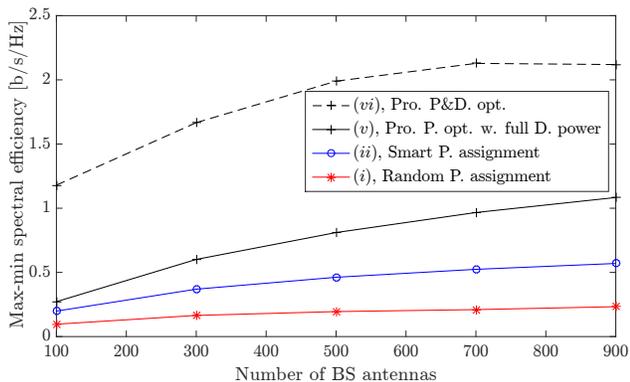} %\vspace*{-0.2cm}
	\caption{ Max-min SE [b/s/Hz] versus the number of BS antennas, $ K = \tau_p = 4 $.}
	\label{Fig-VariousM}
	%\vspace*{-0.2cm}
\end{figure}
\begin{figure}[t]
	\centering
	\includegraphics[trim=1.3cm 8cm 1.8cm 8.5cm, clip=true, width=3.3in]{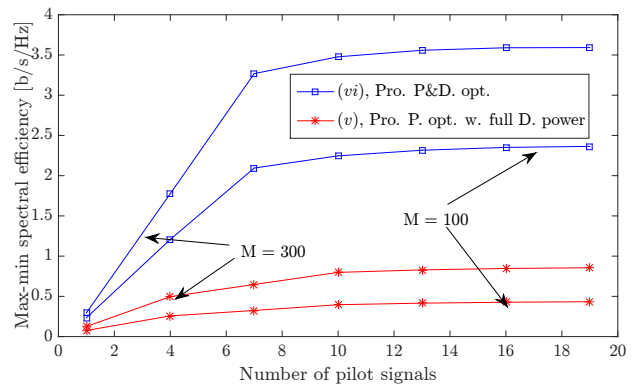} %\vspace*{-0.2cm}
	\caption{ Max-min SE [b/s/Hz] versus the length of pilot signals, $K=4$.}
	\label{Fig-VariousB}
	%\vspace*{-0.2cm}
\end{figure}
\begin{figure}[t]
	\centering
	\includegraphics[trim=1.3cm 8cm 1.8cm 8.5cm, clip=true, width=3.3in]{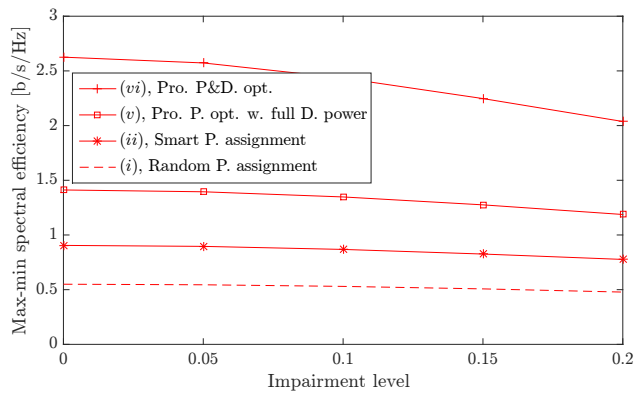} %\vspace*{-0.2cm}
	\caption{ Max-min SE [b/s/Hz] versus the impairment level, $M = 300, K = \tau_p = 2$.}
	\label{Fig-HardwareImpairments}
	%\vspace*{-0.2cm}
\end{figure} 
Fig.~\ref{Fig-CDF-2K2B} shows the cumulative distribution function (CDF) of the max-min SE level [b/s/Hz] for the case $K=\tau_p=2$ and $M =300$. Random pilot assignment yields the worst performance due to the pilot contamination and mutual interference. At the $95\%$-likely point, smart pilot assignment brings significant improvements: it is about $6.86 \times$ better since it attempts to reduce the mutual interference between the users \cite{Xu2015a}. Although the performance of smart pilot assignment is very close to optimal pilot assignment with exhaustive search for a predetermined pilot power level \cite{Xu2015a}, by jointly optimizing power and allocation of the pilot signals, the proposed method outperforms smart pilot assignment by providing an additional $1.53 \times$ gain in average max-min SE. Fig.~\ref{Fig-CDF-2K2B} also demonstrates the superiority of jointly optimizing both data and pilot powers. It yields $1.86 \times$ higher average SE than optimizing only the pilots. Furthermore, the similar performance between the proposed pilot design and exhaustive search confirms the effectiveness of the proposed algorithms for finding locally optimal solutions. The proposed approach for joint optimization pilot and data power on average achieves $84\%$ of the SE with the optimal solution, with a substantially lower computational complexity.

Fig.~\ref{Fig-CDF-Power-2K2B} plots the CDF of power allocated to each pilot/data symbol with $K=\tau_p = 2$ and $M=300$. We do not include random pilot assignment and smart pilot assignment because they consume full power $200$ mW. We observe a vast energy saving for the remaining methods. For the pilot and data power control with exhaustive search, the power consumption on pilot signaling is on average only $111$ mW. The proposed design without data power control, each pilot symbol on average only spends $50$ mW. In comparison to the previous pilot design (e.g., the max-min SE in Fig.~\ref{Fig-CDF-2K2B} and the corresponding power consumption in Fig.~\ref{Fig-CDF-Power-2K2B}), we conclude that the proposed pilot structure can reduce the transmit power more significantly than the prior works while producing better max-min SE. This is because the prior works mitigate pilot contamination by assigning the orthogonal pilot signals to the users; that is, each user assigns power to only one of the basis vectors. In contrast, our design allows for non-orthogonal pilot signals with increasing feasible domain to minimize the total coherent interference in the network. By counting $95$\% of the pilot energy, the average number of non-zero pilot symbols in each pilot signal is on average about $1.02$ and $1.32$ for $\tau_p =K=2$ and $\tau_p =K=4$, respectively. Hence, some users have non-orthogonal pilots.
 \begin{figure}[t]
 	\centering
 	\includegraphics[trim=1cm 7.5cm 1.5cm 8cm, clip=true, width=3.3in]{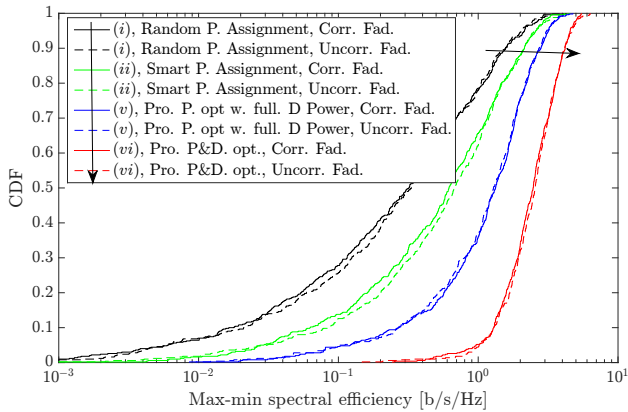} %\vspace*{-0.2cm}
 	\caption{ Max-min SE [b/s/Hz] for the cases of uncorrelated Rayleigh fading and correlated Rayleigh fading, $M=300, K=\tau_p=2$.}
 	\label{Fig-CorrUncorr}
 \end{figure}
The pilot power optimization with exhaustive search has a computational complexity that grows exponentially in $K$. When increasing the number of users, we henceforth only compare the three pilot designs: proposed pilot design, smart pilot assignment, and random pilot assignment. Fig.~\ref{Fig-VariousK} displays the max-min SE as a function of the number of users and pilots with $M=300, K = \tau_p$. Our proposed design yields the highest performance among the related algorithms. Specifically, in comparison to random pilot assignment, the improvement varies from $2.39 \times$ with $K = \tau_p =2$ to $4.58 \times$ with $K = \tau_p = 8 $. Even though smart pilot assignment also performs much better than random pilot assignment, our proposed scheme provides substantial improvements over the state-of-the-art (e.g., we gain $1.72 \times$ with $ K=\tau_p =8$ for fixed data power). The figure also demonstrates a significant improvement with jointly optimizing pilot and data power for all the tested cases of the user and pilot number. In comparison to the case of only optimizing the pilots, the performance gain is from $1.87 \times$ to $4.01 \times$ as the number of users and pilots increases from $2$ to $8$. In addition, we observe a dramatic reduction of the max-min SE of all the pilot designs when the number of users increases due to stronger mutual interference. According to Fig.~\ref{Fig-VariousK}, the performances of the different pilot designs and power control algorithms are about the same for large numbers of users.

Fig.~\ref{Fig-VariousM} shows the max-min SE in the network as a function of the number of BS antennas. Among the three pilot structures, we still observe the worst performance with random pilot assignment, for which the max-min SE only increases from $0.10$ [b/s/Hz] to $0.23$ [b/s/Hz], when the number of antennas increases from $100$ to $900$ antennas. Our proposed pilot design obtains better performance than smart pilot assignment and the gap grows as more antennas BSs are added. For example, with $100$ BS antennas, the smart pilot assignments yields a max-min SE $0.20$ [b/s/Hz], while optimizing pilot signaling for fixed data powers gives $0.27$ [b/s/Hz] and the joint optimization of pilot and data powers yields $1.18$ [b/s/Hz]. The max-min SE is respectively $0.60$, $1.08$, and $2.12$ [b/s/Hz] for the three methods when the BSs are equipped with $900$ antennas. Hence the proposed pilot structure is especially suitable for Massive MIMO systems  with the large ratios $M/K$. Moreover, due to the pilot contamination, a saturation of the SE at about $2.10$ [b/s/Hz] is observed with the proposed joint optimization of pilot and data powers.

Only the special case $K = \tau_p$ has been investigated so far due to the high computational complexity of the pilot assignments from \cite{Xu2015a,Jin2015a}. In contrast, our algorithms can be applied for arbitrary lengths of the pilots. This is demonstrated in Fig.~\ref{Fig-VariousB}, which displays the max-min SE as a function of the length of the pilot signals. We observe noticeable gains when the pilot length increases for both $(v)$ and $(vi)$. For instance, when the pilot length per coherence interval increases from $1$ to $19$, we observe a gain up to $12.17 \times$  with the proposed joint pilot and data power optimization. Moreover, the pilot signals with lengths $\tau_p > KL$ still bring an increasing SE since it produces better channel estimation quality. However the gain is not significant, and therefore $\tau_p = KL$ is a good choice to eliminate pilot contamination and maximize the minimum SE if we have enough time-frequency resources. 

The impact of hardware impairments is shown in Fig.~\ref{Fig-HardwareImpairments}. Increasing the impairment level leads to a decrease in max-min SE, but the reduction is not dramatic since we are operating at low SE where the performance is already interference limited, so the additional self-interference has no dominant impact. When the impairment level reaches $0.2$, there is only a reduction of about $15.04 \%, 16.51\%, 18.88\%,$ and $28.87\%$ for random pilot assignment, smart pilot assignment, proposed pilot design with full data power, and proposed joint pilot and data power optimization respectively. This impairment level is hardly found in practice. For a more realistic impairment level of $\epsilon = 0.1$, the reduction only varies from $3.61\%$ to $7.27\%$. Furthermore that hardware impairments reduce more the performance of the algorithms giving higher SEs. Finally, we stress that our proposed algorithms provide substantial performance gains also under hardware impairments.

Motivated by the fact that practical channels may be spatially correlated, Fig.~\ref{Fig-CorrUncorr} shows the max-min SE of the pilot design (please refer to Section~\ref{Section:CorrUncorr} for the analysis) where the channel between user $i$ in cell $t$ and BS $l$ is modeled by using the exponential correlation model \cite{Loyka2001a}. In particular, $\mathbf{h}_{i,t}^l \sim \mathcal{CN}(\mathbf{0}, \mathbf{R}_{i,t}^l)$ and $\mathbf{R}_{i,t}^l$ is defined as
\begin{equation}
\mathbf{R}_{i,t}^l = \beta_{i,t}^l \begin{bmatrix}
1     &  r_{i,t}^{l,\ast}   &  \cdots     &  (r_{i,t}^{l,\ast})^{M-1} \\
 r_{i,t}^{l}       & 1     & \cdots    & (r_{i,t}^{l,\ast})^{M-2} \\
\vdots & \vdots &  \ddots      &  \vdots  \\
(r_{i,t}^{l})^{M-1}     & (r_{i,t}^{l})^{M-2}    &  \cdots   &  1
\end{bmatrix},
\end{equation}
where the correlation coefficient $r_{i,t}^l =\rho_{i,t}^l e^{\theta_{i,t}^l}$ and $\theta_{i,t}^l$ is the angle between the BS-user vector and the horizontal line. The coefficient magnitude $\rho_{i,t}^l$ is in the range of $[0,1]$ and it equals $0.5$ in our simulation. The figure demonstrates that our pilot design provides much higher max-min SE than the other pilot designs. Although the spatial correlation affects the user channels, the differences in SE are not significant compared to uncorrelated Rayleigh fading.

\section{Conclusion} \label{Section: Conclusion}
%\vspace*{-0.3cm}
This paper proposed a novel pilot design for cellular Massive MIMO systems and combines the pilot assignment and uplink power allocation into a unified optimization framework. A key difference from prior work is that we treat the pilot signals as continuous optimization variables, instead of predefined vectors that should be assigned combinatorially. We used the proposed pilot structure to compute a new SE on the UL ergodic SE for Rayleigh fading channels with MR detection, both with ideal hardware and hardware impairments. A general max-min SE optimization problem was formulated. Although finding the globally optimal solution is NP-hard, we developed an algorithm that finds a local optimum that outperforms the previous state-of-the-art algorithms for pilot assignment and is close to the optimum obtained by the exhaustive search. The numerical results demonstrate the importance of jointly optimizing the pilots and transmit powers in order to improve the max-min SE.
\appendix
%\vspace*{-0.5cm}
\subsection{Proof of Lemma~\ref{Lemma: Achievable_Rate}} \label{Appendix: Proof-Achievable_Rate}
%\vspace*{-0.3 cm}
The SINR value in the theorem is obtained by computing the expectations in \eqref{eq: SINR_k}. Due to the independence of the channel estimate and estimation error, the numerator of \eqref{eq: SINR_k} is computed as
\begin{equation} \label{eq: NumeratorProof}
\begin{split}
&p_{l,k} |\mathbb{E} \{ \mathbf{v}_{l,k}^H \mathbf{h}_{l,k}^l \} |^2 = p_{l,k} |\mathbb{E} \{ || \hat{\mathbf{h}}_{l,k}^l ||^2 \} |^2 \\
&=  p_{l,k} \left( M \beta_{l,k}^l \varrho_{l,k}^l \sum_{b=1}^{\tau_p} \hat{p}_{l,k}^b \right)^2.
\end{split}
\end{equation}
Note that this term also appears in the denominator of \eqref{eq: SINR_k}. The first expectation in  the denominator is reformulated as
\begin{equation} \label{eq: FirstTermMRT}
\begin{split}
 &\mathbb{E}\{ | \mathbf{v}_{l,k}^{H} \mathbf{h}_{i,t}^{l} |^2 \} / (\varrho_{l,k}^l)^2 =  \mathbb{E} \left\{ | \mathbf{y}_{l,k}^H \mathbf{h}_{i,t}^{l} |^2 \right\} \\
& =  \mathbb{E} \left\{ \left| \sum_{(i',t') \in \mathcal{S}}  \sum_{b=1}^{\tau_p} \sqrt{\hat{p}_{l,k}^b \hat{p}_{i',t'}^b } (\mathbf{h}_{i',t'}^{l})^H \mathbf{h}_{i,t}^l \right|^2 \right\} \\
& \; \; \; + \mathbb{E} \left\{ \left| \left( \mathbf{N}_{l} \pmb{\psi}_{l,k} \right)^H \mathbf{h}_{i,t}^l \right|^2 \right\}
 \end{split}
\end{equation}
by utilizing the independence of the noise and the channels.
Decomposing the first expectation into two parts, one related to $\mathbf{h}_{i,t}^l$ and the other comprises of the remaining channels, we obtain
\begin{equation} \label{eq:I1}
\begin{split}
&  \mathbb{E} \left\{ \left| \sum_{(i',t') \in \mathcal{S}}  \sum_{b=1}^{\tau_p} \sqrt{\hat{p}_{l,k}^b \hat{p}_{i',t'}^b } (\mathbf{h}_{i',t'}^{l})^H \mathbf{h}_{i,t}^l \right|^2 \right\} \\ 
&  =  \mathbb{E} \left\{ \left| \sum_{b=1}^{\tau_p} \sqrt{\hat{p}_{l,k}^b \hat{p}_{i,t}^b} (\mathbf{h}_{i,t}^{l})^H \mathbf{h}_{i,t}^l \right|^2 \right\} \\
&\; \; \; +  \mathbb{E} \left\{ \left| \sum_{(i',t') \in \mathcal{S}_{i,t} \setminus (i,t)}  \sum_{b=1}^{\tau_p} \sqrt{\hat{p}_{l,k}^b \hat{p}_{i',t'}^b} (\mathbf{h}_{i',t'}^{l})^H \mathbf{h}_{i,t}^l \right|^2 \right\} \\ 
& \stackrel{(a)}{=} M(M+1) \left(\sum_{b=1}^{\tau_p} \sqrt{\hat{p}_{l,k}^b \hat{p}_{i,t}^b} \right)^2 (\beta_{i,t}^l)^2 \\
&\; \; \; + M \mathlarger{\sum_{ (i',t') \in \mathcal{S}_{i,t} \setminus (i,t) } } \left( \sum_{b=1}^{\tau_p} \sqrt{ \hat{p}_{l,k}^b  \hat{p}_{i',t'}^b } \right)^2 \beta_{i',t'}^l \beta_{i,t}^l \\
&=  M^2 \left(\sum_{b=1}^{\tau_p} \sqrt{\hat{p}_{l,k}^b \hat{p}_{i,t}^b} \right)^2 (\beta_{i,t}^l)^2\\
&\; \; \; + M \sum_{(i',t') \in \mathcal{S} }\left( \sum_{b=1}^{\tau_p} \sqrt{ \hat{p}_{l,k}^b  \hat{p}_{i',t'}^b } \right)^2 \beta_{i',t'}^l \beta_{i,t}^l, 
\end{split}
\end{equation}
where $(a)$ is obtained by using Lemma $2.9$ in \cite{Tulino2004} to compute the fourth-order moment $\mathbb{E} \{ || \mathbf{h}_{i,t}^l ||^4\} $. The second expectation in \eqref{eq: FirstTermMRT} is computed from the independence property between additive noise and the original channel $\mathbf{h}_{i,t}^l$ as
\begin{equation} \label{eq: Part2ofFirstDe}
\mathbb{E} \{ | \left( \mathbf{N}_{l} \pmb{\psi}_{l,k} \right)^H \mathbf{h}_{i,t}^l |^2 \} = M \beta_{i,t}^l \sigma^2 \sum_{b=1}^{\tau_p} \hat{p}_{l,k}^b.
\end{equation}
Plugging \eqref{eq:I1} and \eqref{eq: Part2ofFirstDe} into \eqref{eq: FirstTermMRT}, the first expectation in the denominator of \eqref{eq: SINR_k} is
\begin{equation} \label{eq:FirstDenominator}
\begin{split}
& \mathbb{E}\{ | \mathbf{v}_{l,k}^{H} \mathbf{h}_{i,t}^{l} |^2 \} = M (\varrho_{l,k}^l)^2 \sum_{(i,t) \in \mathcal{S}} p_{i,t}  \beta_{i,t}^l \times\\
&\left( \frac{\beta_{l,k}^l \sum_{b=1}^{\tau_p} \hat{p}_{l,k}^b }{\varrho_{l,k}^l} +  M \left(\sum_{b=1}^{\tau_p} \sqrt{\hat{p}_{l,k}^b \hat{p}_{i,t}^b} \right)^2  \beta_{i,t}^l  \right).
 \end{split}
\end{equation}
The last expectation in the denominator of \eqref{eq: SINR_k} is  
\begin{equation} \label{eq: lastExpectation}
\mathbb{E} \{ \| \mathbf{v}_{l,k} \|^2 \} = M \beta_{l,k}^l \varrho_{l,k}^l \sum_{b=1}^{\tau_p} \hat{p}_{l,k}^b,
\end{equation}
by using the definition of MR detection in \eqref{eq:MRDetection} and the estimated channel distribution in Lemma~\ref{lemma: Distribution}. Plugging \eqref{eq: NumeratorProof}, \eqref{eq:FirstDenominator}, and \eqref{eq: lastExpectation} into \eqref{eq: SINR_k}, we obtain the SINR value as shown in the theorem after some simple algebra.

\setcounter{eqnback}{\value{equation}} \setcounter{equation}{84}
\begin{figure*}
	\begin{equation} \label{eq: GeneralHardwareRate}
	\overline{\mathrm{SINR}}_{l,k}^{\textrm{MR}} 
	=  \frac{ (1 - \epsilon^2) p_{l,k} | \mathbb{E} \{ \mathbf{y}_{l,k}^{H} \mathbf{h}_{l,k}^l  \} |^2}{ \sum\limits_{(i,t) \in \mathcal{S} } p_{i,t} \mathbb{E} \{ | \mathbf{y}_{l,k}^{H} \mathbf{h}_{i,t}^{l} |^2 \}  - (1 - \epsilon^2) p_{l,k} | \mathbb{E} \{ \mathbf{y}_{l,k}^{H} \mathbf{h}_{l,k}^{l} \} |^2 + \sigma^2 \mathbb{E} \{ \| \mathbf{y}_{l,k} \|^2 \}}.
	\end{equation}
	\hrulefill
\end{figure*}
\setcounter{eqncnt}{\value{equation}}
\setcounter{equation}{\value{eqnback}}

\subsection{Proof of Theorem~\ref{Theorem: KKTpoint}} \label{Appendix: Proof-Theorem:KKTPoint}
%\vspace*{-0.3 cm}
 Let us denote  the feasible set of the optimization problem \eqref{eq: Opt_Prob2} as
 \begin{equation}
 \begin{split}
 \mathcal{F} = \Big\{ &p_{l,k}, \hat{p}_{l,k}^b, \forall l,k,b:\; p_{l,k},\hat{p}_{l,k}^b \in \mathbb{R}_{+} ,  p_{l,k} \leq P_{\max,l,k}^d, \\
 &\frac{1}{\tau_p} \sum_{b=1}^{\tau_p} \hat{p}_{l,k}^b \leq P_{\max,l,k} \Big\}.
 \end{split}
 \end{equation}
The optimal solution set to the optimization problem \eqref{eq: Opt_Prob3} in the $i$th iteration is denoted as
\begin{equation}
 \mathcal{I}^{ (i)}= \left\{  p_{l,k}^{(i),\mathrm{opt}}, \hat{p}_{l,k}^{b,(i),\mathrm{opt}}, \forall l,k,b \right\}.
\end{equation}
 Let $\mathrm{SINR}_{l,k}^{ \textrm{MR}} (\tilde{f})$ denote the SINR value computed from \eqref{eq: SINR_MRC1} for any feasible point $\tilde{f} \in \mathcal{F}$, while $\doublewidetilde{\mathrm{SINR}}_{l,k}^{(i), \textrm{MR}} (\mathcal{I}^{(i)})$ denotes the approximated SINR value at the $i$th iteration computed from \eqref{eq: SINR_MRCApproximation} using the solution $\mathcal{I}^{(i)}$ obtained in the $i$th iteration. By using the approximation in \eqref{eq: SINRBound}, our family of SINR functions satisfies the following properties \cite{Chiang2007b}, $\forall l,k,b$:
 \begin{align}
 \mathrm{SINR}_{l,k}^{\textrm{MR}} \left(\tilde{f} \right) & \geq \doublewidetilde{\mathrm{SINR}}^{ (i),\textrm{MR}}_{l,k} \left(\tilde{f}\right), \forall \tilde{f} \in \mathcal{F}, \label{eq:SINRCond1} \\
 \mathrm{SINR}_{l,k}^{\textrm{MR}} \left(\mathcal{I}^{(i)}\right) & = \doublewidetilde{\mathrm{SINR}}^{(i+1),\textrm{MR}}_{l,k} \left(\mathcal{I}^{(i)}\right), \label{eq:SINRCond2} \\
 \frac{\partial \mathrm{SINR}_{l,k}^{\textrm{MR}} \left(\mathcal{I}^{(i)}\right) }{ \partial \hat{p}_{l,k}^{b} } & = \frac{\partial \doublewidetilde{\mathrm{SINR}}^{(i+1),\textrm{MR}}_{l,k} \left(\mathcal{I}^{(i)}\right)}{  \partial \hat{p}_{l,k}^{b} },
  \label{eq:SINRCond3}\\
\frac{\partial \mathrm{SINR}_{l,k}^{\textrm{MR}} \left(\mathcal{I}^{(i)}\right) }{ \partial p_{l,k} } & = \frac{\partial \doublewidetilde{\mathrm{SINR}}^{(i+1),\textrm{MR}}_{l,k} \left(\mathcal{I}^{(i)}\right)}{  \partial p_{l,k} }. \label{eq:SINRCond4}
 \end{align}
 The property \eqref{eq:SINRCond1} implies that the globally objective value to the geometric program \eqref{eq: Opt_Prob3} is also feasible for the signomial program \eqref{eq: Opt_Prob2} and we can construct the following chain of inequalities:
 \begin{equation}
 \begin{split}
 \ldots &=\doublewidetilde{\mathrm{SINR}}^{(i),\textrm{MR}}_{l,k} \left( \mathcal{I}^{(i-1)}\right) \overset{(a)}{\leq}  \doublewidetilde{\mathrm{SINR}}^{(i),\textrm{MR}}_{l,k} \left( \mathcal{I}^{(i)}\right) \\
 &
   \overset{(b)}{\leq} \mathrm{SINR}^{\textrm{MR}}_{l,k} \left( \mathcal{I}^{(i)}\right) \overset{(c)}{=} \doublewidetilde{\mathrm{SINR}}^{(i+1),\textrm{MR}}_{l,k} \left( \mathcal{I}^{(i)}\right) \leq \ldots,
  \end{split}
 \end{equation}
 where $(a)$ is obtained by solving the geometric optimization problem \eqref{eq: Opt_Prob3}. $(b)$ and $(c)$ follow by \eqref{eq:SINRCond1} and \eqref{eq:SINRCond2}, respectively. Thus, if $\xi^{(i),\mathrm{opt}}$ is the optimal objective value of \eqref{eq: Opt_Prob3}, then we obtain $\xi^{ (i+1),\mathrm{opt}} \geq \xi^{(i),\mathrm{opt}}$. The objective function is non-decreasing with the iteration index $i$, while the pilot and data power ranges make the SINR expressions continuous functions and they are bounded from above (i.e., $\doublewidetilde{\mathrm{SINR}}_{l,k}^{\textrm{MR}} < \infty, \forall l,k$) ensuring that \eqref{eq: Opt_Prob3} converges. If the convergence holds at the $i$th iteration (i.e., $\xi^{(i+1),\mathrm{opt}} = \xi^{(i),\mathrm{opt}}$), then the optimal solution set $\mathcal{I}^{ (i)}$ must also be a solution in the $(i+1)$th iteration (otherwise, it leads to $\xi^{ (i+1),\mathrm{opt}} > \xi^{(i),\mathrm{opt}}$). Hence \eqref{eq: Opt_Prob3} converges to a limit point. Furthermore, our constraint functions in \eqref{eq: Opt_Prob3} satisfy Slater's condition \cite{Boyd2004a} and ensure that the KKT conditions of \eqref{eq: Opt_Prob2} and \eqref{eq: Opt_Prob3} coincide, if we use \eqref{eq:SINRCond3} and \eqref{eq:SINRCond4} to do a matching procedure as in the proof of Theorem $1$ in \cite{Marques1978a}. Consequently, the limit point obtained when solving \eqref{eq: Opt_Prob3} in an iterative manner is a KKT local point to \eqref{eq: Opt_Prob2}.

\subsection{Proof of Theorem ~\ref{Theorem: ImperfectHWRate}} \label{Appendix: Proof-Theorem:ImperfectHWRate}
%\vspace*{-0.3 cm}

We use Lemma $4$ \cite{Bjornson2016a} to obtain the effective SINR of user~$k$ in cell~$l$ with MR detection as in \eqref{eq: GeneralHardwareRate}. It remains to compute the expectations, which is similar to proof of Lemma~\ref{Lemma: Achievable_Rate}, but more complicated since the channel estimates and estimation errors are neither Gaussian nor independent. Therefore, we summarize the main steps to compute these expectations.

By utilizing the uncorrelation property among the noise, the distortion and the propagation channels, the numerator of \eqref{eq: GeneralHardwareRate} is computed as
\setcounter{eqnback}{\value{equation}} \setcounter{equation}{85}
\begin{equation} \label{eq: NumeratorHW}
\begin{split}
& (1 - \epsilon^2) p_{l,k} | \mathbb{E} \{ \mathbf{y}_{l,k}^{H} \mathbf{h}_{l,k}^l  \} |^2 = (1 - \epsilon^2)^2 p_{l,k} M^2 (\beta_{l,k}^l)^2 \| \pmb{\psi}_{l,k} \|^4 \\
 &= (1 - \epsilon^2)^2 p_{l,k} M^2 (\beta_{l,k}^l)^2 \left( \sum_{b=1}^{\tau_p} \hat{p}_{l,k}^b\right)^2.
 \end{split}
\end{equation}
Similarly, thanks to the uncorrelation of the noise, the propagation channels, and the distortion term, the first expectation in  the denominator of \eqref{eq: GeneralHardwareRate} is computed as
\begin{equation} \label{eq: FirstExpectation}
\begin{split}
 &\mathbb{E} \{ | \mathbf{y}_{l,k}^H \mathbf{h}_{i,t}^l |^2 \}  =\\
 & M \beta_{i,t}^l \left( \sum_{(i',t') \in \mathcal{S}} \beta_{i',t'}^l \kappa_{i',t'} + \| \pmb{\psi}_{l,k} \|^2 \sigma^2 \right) + M^2 \kappa_{i,t} (\beta_{i,t}^l)^2,
 \end{split}
\end{equation}
where $\kappa_{i,t} =  (1 - \epsilon^2) | \pmb{\psi}_{i,t}^H \pmb{\psi}_{l,k}|^2 + \pmb{\psi}_{l,k}^H \pmb{\Lambda}_{i,t} \pmb{\psi}_{l,k}$, and then by utilizing \eqref{eq: Orthogonal_Property} we obtain the alternative expression for $\kappa_{i,t}$ that is shown in \eqref{eq: Kappa_it}. In the same manner, the last expectation in the denominator of \eqref{eq: GeneralHardwareRate} is
\begin{equation} \label{eq:LastExpectation}
\mathbb{E} \{ \| \mathbf{y}_{l,k} \|^2 \} = M \left( \sum_{(i,t) \in \mathcal{S} } \beta_{i,t}^l \kappa_{i,t} + \sigma^2 \| \pmb{\psi}_{l,k} \|^2 \right).
\end{equation}
Plugging  \eqref{eq: NumeratorHW}--\eqref{eq:LastExpectation} into \eqref{eq: GeneralHardwareRate}, we obtain the SINR given in the theorem by using the following identities:
\begin{align}
\| \pmb{\psi}_{l,k} \|^4 &= \left( \sum_{b=1}^{\tau_p} \hat{p}_{l,k}^b \right)^2 \\
| \pmb{\psi}_{i,t}^H \pmb{\psi}_{l,k}|^2 & = \left( \sum_{b=1}^{\tau_p} \sqrt{\hat{p}_{i,t}^b \hat{p}_{l,k}^b}  \right)^2 \\
 \pmb{\psi}_{l,k}^H  \pmb{\Lambda}_{i,t} \pmb{\psi}_{l,k} & =  \epsilon^2 \sum_{b=1}^{\tau_p} \hat{p}_{i,t}^b \hat{p}_{l,k}^b.
\end{align}

%\vspace*{-0.4 cm}

%==========================Reference==========================================================================
\bibliographystyle{IEEEtran}
\bibliography{IEEEabrv,refs}

\begin{IEEEbiography} 
	%[{\includegraphics[width=1.0in,height=1.25in,clip,keepaspectratio]{trinhvanchien.jpg}}]
	{Trinh Van Chien} (S'16) received the B.S. degree in Electronics and Telecommunications from Hanoi University of Science and Technology (HUST), Vietnam, in 2012. He then received the M.S. degree in Electrical and Computer Enginneering from Sungkyunkwan University (SKKU), Korea, in 2014. He is currently working towards the Ph.D. degree in Communication Systems from Link\"oping University (LiU), Sweden. His interest lies in convex optimization problems for wireless communications and image \& video processing. He was an IEEE wireless communications letters exemplary reviewer for 2016. He also received the award of scientific excellence in the first year of the 5Gwireless project funded by European Union Horizon's 2020.
\end{IEEEbiography}

\begin{IEEEbiography}
	%[{\includegraphics[width=1.0in,height=1.25in,clip,keepaspectratio]{emilbjornson.jpg}}]
	{Emil Bj\"{o}rnson}(S'07, M'12, SM'17) received the M.S. degree in Engineering Mathematics from Lund University, Sweden, in 2007. He received the Ph.D. degree in Telecommunications from KTH Royal Institute of Technology, Sweden, in 2011. From 2012 to mid 2014, he was a joint postdoc at the Alcatel-Lucent Chair on Flexible Radio, SUPELEC, France, and at KTH. He joined Linköping University, Sweden, in 2014 and is currently Associate Professor and Docent at the Division of Communication Systems.
	
	He performs research on multi-antenna communications, Massive MIMO, radio resource allocation, energy-efficient communications, and network design. He is on the editorial board of the \textsc{IEEE Transactions on Communications} and the \textsc{IEEE Transactions on Green Communications and Networking}. He is the first author of the textbooks \emph{Massive MIMO Networks: Spectral, Energy, and Hardware Efficiency} (2017) and \emph{Optimal Resource Allocation in Coordinated Multi-Cell Systems} (2013). He is dedicated to reproducible research and has made a large amount of simulation code publicly available.
	
	Dr. Bj\"{o}rnson has performed MIMO research for more than ten years and has filed more than ten related patent applications. He received the 2016 Best PhD Award from EURASIP, the 2015 Ingvar Carlsson Award, and the 2014 Outstanding Young Researcher Award from IEEE ComSoc EMEA. He has co-authored papers that received best paper awards at WCSP 2017, IEEE ICC 2015, IEEE WCNC 2014, IEEE SAM 2014, IEEE CAMSAP 2011, and WCSP 2009.
\end{IEEEbiography}
\begin{IEEEbiography}
	%[{\includegraphics[width=1.0in,height=1.25in,clip,keepaspectratio]{erikglarsson.jpeg}}]
	{Erik G. Larsson} (S'99, M'03, SM'10, F'16) received the Ph.D. degree from Uppsala University,
	Uppsala, Sweden, in 2002. 
	
	He is currently Professor of Communication Systems at Linköping University (LiU) in Linköping, Sweden. He was with the Royal Institute of Technology (KTH) in Stockholm, Sweden, the University of Florida, USA, the George Washington University, USA, and Ericsson Research, Sweden. In 2015 he was a Visiting Fellow at Princeton University, USA, for four months. His main professional interests are within the areas of wireless communications and signal processing. He has co-authored some 130 journal papers on these topics, he is co-author of the two Cambridge University Press textbooks Space-Time Block Coding for Wireless Communications (2003) and Fundamentals of Massive MIMO (2016). He is co-inventor on 16 issued and many pending patents on wireless technology. 
	
	He is a member of the IEEE Signal Processing Society Awards Board during 2017–2019. From 2015 to 2016 he served as chair of the IEEE Signal Processing Society SPCOM technical committee. From 2014 to 2015 he 1246
	was chair of the steering committee for the IEEE Wireless Communications Letters. He was the General Chair of the Asilomar Conference on Signals, Systems and Computers in 2015, and its Technical Chair in 2012. He was Associate Editor for, among others,the IEEE Transactions on Communications (2010-2014) and the IEEE Transactions on Signal Processing (2006-2010). He received the IEEE Signal Processing Magazine Best Column Award twice, in 2012 and 2014, the IEEE ComSoc Stephen O. Rice Prize in Communications Theory in 2015, and the IEEE ComSoc Leonard G. Abraham Prize in 2017.
\end{IEEEbiography}

\end{document}